\documentclass[11pt,reqno]{amsart}
\usepackage{amsmath,amssymb}
\usepackage{amsthm} 
\usepackage{amscd}

\usepackage[pdftex]{graphicx}
\usepackage[all]{xy}

\usepackage{color}
\newcommand{\revs}[1]{{#1}}

\makeatletter
\@addtoreset{equation}{section}
\makeatother
\usepackage{enumerate}

\def\be{\mathbf e}

\def\CC{{\mathbb C}}

\def\LL{{\mathbb L}}
\def\RR{{\mathbb R}}

\def\PP{{\mathbb P}}

\def\be{{\mathbf e}}

\def\ee{{\mathrm e}}

\def\ii{{\sqrt{-1}}}

\def\cC{{\mathcal C}}
\def\cD{{\mathcal D}}

\def\cG{{\mathcal G}}

\def\cK{{\mathcal K}}
\def\cL{{\mathcal L}}
\def\cM{{\mathcal M}}

\def\cO{{\mathcal O}}

\def\cV{{\mathcal V}}

\def\fL{{\mathfrak L}}

\def\fX{{\mathfrak X}}

\def\fa{{\mathfrak a}}
\def\fb{{\mathfrak b}}
\def\fc{{\mathfrak c}}

\def\ft{{\mathfrak t}}

\def\ri{{\mathrm i}}
\def\rr{{\mathrm r}}

\def\tA{{\widetilde A}}

\def\tE{{\widetilde E}}

\def\tK{{\widetilde K}}

\def\tX{{\widetilde X}}

\def\tS{{\widetilde S}}

\def\tX{{\widetilde X}}

\def\tv{{\widetilde v}}
\def\tx{{\widetilde x}}

\def\tgamma{{\widetilde \gamma}}

\def\tvarphi{{\widetilde \varphi}}

\def\tfS{{\widetilde{\mathfrak S}}}

\def\tcV{{\widetilde{\mathcal V}}}
\def\tSS{{{\widetilde{\mathbb S}}}}

\def\hv{{\widehat v}}

\def\hB{{\widehat B}}

\def\hJ{{\widehat J}}

\def\hX{{\widehat X}}

\def\hkappa{{\widehat \kappa}}

\def\hnu{{\widehat \nu}}

\def\hvarpi{{\widehat \varpi}}

\def\hGamma{{\widehat \Gamma}}

\def\cn{{\mathrm {cn}}}
\def\sn{{\mathrm {sn}}}
\def\dn{{\mathrm {dn}}}

\def\al{{\mathrm {al}}}

\def\qed{\hbox{\vrule height6pt width3pt depth0pt}}
\def\nuI#1{{\nu_{#1}}}

\def\trp{{\, {}^t\negthinspace}}

\def\book#1{\rm{#1}, }
\def\paper#1{\textit{#1}, }
\def\jour#1{\rm{#1}, }
\def\yr#1{({\rm{#1}) }}
\def\vol#1{\textbf{#1}}
\def\pages#1{\rm{#1}}

\def\publaddr#1{\rm{#1}, }
\def\publ#1{\rm{#1}, }
\def\by#1{{\rm{#1}, }}

\newtheorem{theorem}{Theorem}[section]

\newtheorem{proposition}[theorem]{Proposition}

\newtheorem{remark}[theorem]{Remark}
\newtheorem{lemma}[theorem]{Lemma}
\newtheorem{assumption}[theorem]{Assumption}

\newtheorem{condition}[theorem]{Condition}

\def\book#1{\rm{#1}, }
\def\paper#1{\textit{#1}, }
\def\jour#1{\rm{#1}, }
\def\yr#1{({\rm{#1}) }}
\def\vol#1{\textbf{#1}}
\def\pages#1{\rm{#1}}

\def\publaddr#1{\rm{#1}, }
\def\publ#1{\rm{#1}, }
\def\by#1{{\rm{#1}, }}

\begin{document}


\title[On real hyperelliptic solutions of focusing modified KdV equation]{On real hyperelliptic solutions of focusing modified KdV equation}

\author{Shigeki Matsutani}
%

\date{\today}

\begin{abstract}
We study the real hyperelliptic solutions of the focusing modified KdV (MKdV) equation of the genus three.
Since the complex hyperelliptic solutions of the focusing MKdV equation over $\CC$ are associated with the real gauged MKdV equation, we present a novel construction of the real hyperelliptic solutions of the gauged MKdV equation.
When the gauge field is constant, it can be regarded as the real solution of the focusing MKdV equation, and thus we also discuss the behavior of the gauge field numerically.
\end{abstract}

\keywords{modified KdV equation, real hyperelliptic solutions, hyperelliptic curves, focusing MKdV equation}



\maketitle
\section{Introduction}\label{sec:1}

\revs{
It is known that the modified Korteweg-de Vries (MKdV) equation is given by
\begin{equation}
\partial_t q \pm 6q^2 \partial_s q +\partial_s^3 q =0,
\label{eq:MKdV1}
\end{equation}
where $\partial_t := \partial/\partial t$, $\partial_s := \partial/\partial s$, and $t$ and $s$ are the real axes \cite{AS}.
Here, the $+$ case in (\ref{eq:MKdV1}) is called the focusing MKdV equation, while $-$ is called the defocusing MKdV equation \cite{ZakharovShabat}.
Although the MKdV equation is connected with geometry as a helical curve by G. Lamb Jr \cite{Lamb} related to the nonlinear Schr\"odinger equation, the study directly related to this paper was opened by Konno, Ichikawa and Wadati \cite{KIW}.
Konno, Ichikawa and Wadati found a novel geometrical integrable differential equation by studying the inverse scattering method \cite{GGKM}, whose stationary solution reproduced the shape of Euler's elastica \cite{Euler44}.
They called its soliton solution \emph{loop soliton}, and thus the equation is  also referred to \emph{loop soliton equation}.
Following their study, Ishimori revealed the connection between the loop soliton equation and the focusing MKdV equation \cite{Ishimori, Ishimori2}.
The loop soliton equation is regarded as the equation for the deformation of a real curve whose tangential angle $\theta$ obeys the equation,
\begin{equation}
\partial_{t}\theta
           +\frac{1}{2}
\left(\partial_{s} \theta \right)^3
+\partial_{s}^3 \theta=0.
\label{4eq:rMKdV2_Ishi}
\end{equation}
Ishimori referred to (\ref{4eq:rMKdV2_Ishi}) as the MKdV equation because $q=\partial_s \theta/2$ obviously obeys the MKdV equation (\ref{eq:MKdV1}).
}

\revs{
It is obvious that for the solution $\theta(s,t)$ of (\ref{4eq:rMKdV2_Ishi}), $\phi(s,t):=\theta(s-at, t)$ obeys 
\begin{equation}
(\partial_{t}+\alpha \partial_s)\phi
           +\frac{1}{8}
\left(\partial_{s} \phi \right)^3
+\frac{1}{4}\partial_{s}^3 \phi=0,
\label{4eq:rMKdV2}
\end{equation}
where $\alpha$ is a real parameter.
(\ref{4eq:rMKdV2}) is essentially the same as (\ref{4eq:rMKdV2_Ishi}).
In this paper we refer to (\ref{4eq:rMKdV2}) as the focusing pre-MKdV equation.
}

\revs{
Independently, Goldstein and Petrich found that the iso-metric deformation of a real curve on a plane is connected with the recursion relations of the focusing MKdV hierarchy, and the curve whose curvature obeys the focusing MKdV hierarchy is an iso-metric deformation \cite{GoldsteinPetrich1}.
In other words, the loop soliton equation was rediscovered as an iso-metric deformation of curves on a plane, although it could be seen as a revival of the study \cite{Lamb}.
Previato investigated the iso-metric deformation from points of view of the study of the integrable system \cite{P}, and the author found that the Goldstein-Petrich scheme \cite{GoldsteinPetrich1} plays a crucial role in the statistical mechanics of the elastiae \cite{Mat97}.
Since the plane curve obeying (\ref{4eq:rMKdV2_Ishi}) is regarded as a generalization of Euler's elastica and the partition function in the statistical mechanics is interpreted as the functional integral in the quantum field theory, we sometimes called it quantized elastica \cite{Mat10, M24a}.
(The origin of the deformation was in a crucial problem of the compact Riemann surface \cite{Tj}  and a deformation quantization \cite{Brylinski} as mentioned in \cite{MP16}.)
} 

\revs{
As in \cite{MP22, M24a, Mat97, Mat02b}, to find the hyperelliptic solutions of the focusing pre-MKdV equation is crucial to reproduce the shapes of supercoiled DNA in observed in laboratories.
It means a fascinating and beautiful relationship between modern mathematics and life sciences.
The author and Emma Previato decided to solve this problem (the generalized elastica problem) in 2004 based on the papers \cite{Mat02b, P}.
The author has developed the tools together with co-authors \cite{Mat02b, MP15, MP22, EEMOP08, EHKKLS, EMO08, KMP22} based on Baker's approach and a similar movement \cite{Baker97, BEL97b, BEL20, BuMi2}.
To solve the problem, they considered that a novel approach was needed that directly connects the algebraic curves, rather than the theta-function approach \cite{P0}.
They have spent two decades refining and reconstructing Abelian function theory, including hyperelliptic function theory, as problems in algebraic geometry \cite{Mat10, MP22, M24}.
}

In the previous paper \cite{MP22}, the author with Emma Previato investigated hyperelliptic representations of loop solitons, i.e., the real transcendental curves on a plane whose tangential angles obey the focusing \revs{pre-}MKdV equation over the real number field $\RR$.
The purpose of this paper is to proceed the project.

For a hyperelliptic curve $X_g$ given by $y^2 \revs{+} (x-b_0)(x-b_1)\cdots(x-b_{2g})=0$ for $b_i \in \CC$, due to Baker \cite{Baker03, BEL97b, BuMi2, Mat01a}, we find that the hyperelliptic solution of KdV equation as $\wp_{gg}(u)=x_1+ \cdots +x_g$ for $(x_1, \ldots, x_g) \in S^g X_g$ ($g$-th symmetric product of $X_g$) as a function of $u \in \CC^g$ by the \revs{Abel-Jacobi map $v: S^g X_g \to J_X$ for the Jacobi variety $J_X$}, $u=v((x_1,y_2), \ldots, (x_g, y_g))$, i.e., 
\begin{equation}
   4\lambda_{2g} \partial_{u_g} \wp_{gg} + 4 \partial_{u_{g-1}} \wp_{gg}
\revs{+}\partial_{u_g}^3\wp_{gg}\revs{-}12\wp_{gg}\partial_{u_g}\wp_{gg}\revs{=0},
\label{eq:KdV1}
\end{equation}
where $\lambda_{2g}=-\displaystyle{\sum_{i=0}^{2g} b_i}$ is a constant, and $\partial_{u_i}:=\partial/\partial u_i$.
\revs{
Here we have used the curve $y^2=-f(x)$ instead of $z^2=f(x)$ for later convenience, but the two are simply connected by $y=\ii z$ and thus the parameter in the Jacobian is also given by the multiplication of $\ii$.
}
Using the similar consideration, we can also find the hyperelliptic solution of the MKdV equation over $\CC$ \cite{Mat02b}.

The real hyperelliptic solutions of KdV equation has been studied well \cite{BBEIM}.
Roughly speaking, we consider a compact Riemann surface \revs{$Y_g$} of genus $g$ with an anti-holomorphic involution \revs{$\tau_A: Y_g\to Y_g$}, $\tau_A^2 =1$, such that the fixed point of $\tau_A$ forms $g+1$ ovals; 
such \revs{$Y_g$} is called an M-curve.
Since the KdV equation is connected with hyperelliptic curves as in (\ref{eq:KdV1}), we deal with the M-curve with hyperelliptic involution, which is given by
\begin{equation}
\revs{Y_g^\RR}: y^2 \revs{+} (x-a_0)(x-a_1)\cdots(x-a_{2g})\revs{=0}, 
\label{eq:XgR}
\end{equation}
where $a_i$ are mutually disjoint real values, i.e. $a_j \in \RR$.
By restricting the points $((x_1,y_2), \ldots, (x_g, y_g))$ in $S^g \revs{Y_g^\RR}$ such that $x_i \in \RR$ and $u\in \RR^g$, it is easily to find the real hyperelliptic solutions of the KdV equation due to (\ref{eq:KdV1})

With the help of the Miura map it is not difficult to find the real hyperelliptic solutions of the defocusing MKdV equation \cite{Mat02c},\revs{
\begin{equation}
\partial_{u_{g-1}}\mu
\revs{+}\frac{1}{2}(\lambda_{2g}+3b_0) \partial_{u_g}\mu
-\frac{1}{2}
\left(
\mu^2 \partial_{u_g} \mu\right)^3
+\frac{1}{4}\partial_{u_g}^3 \mu = 0,
\label{1eq:drMKdV_}
\end{equation}
and the defocusing pre-MKdV equation
\begin{equation}
\partial_{u_{g-1}}\xi
\revs{+}\frac{1}{2}(\lambda_{2g}+3b_0) \partial_{u_g}\xi
-\frac{1}{8}
\left(
\partial_{u_g} \xi\right)^3
+\frac{1}{4}\partial_{u_g}^3 \xi = 0,
\label{1eq:dprMKdVxi}
\end{equation}
where $\mu := \partial_{u_g} \xi/2$ and $\xi:=\log( (b_0-x_1)\cdots(b_0-x_g))$.
In this case, $\xi$ and $\mu$ are real valued for $(x_1, \ldots, x_g) \in S^g \RR$ and $u\in \RR^g$.
}

\revs{
However, (\ref{1eq:dprMKdVxi}) is not a real solutions of the loop soliton equation (\ref{4eq:rMKdV2}).}
It roughly means that \revs{$\xi$} in (\ref{1eq:dprMKdVxi}) must be purely imaginary and $u \in \RR^g$.
The conditions that \revs{$\xi$}$\in \ii \RR$ and $u\in \RR^g$ cannot be handled in the above framework.
The former condition is satisfied by the condition that $(b_0-x_1), \ldots, (b_0-x_g)$ draws an arc of a unit circle in a hyperelliptic curve $X_g$ whose center is one of the branch points in $X_g$.
The reality condition that $\phi$ in (\ref{4eq:rMKdV2}) is real causes the transcendental transformation, i.e., the logarithmic transformation of $S^1 \to \RR$.
Accordingly, we cannot find a real hyperelliptic curve \revs{$Y_g^\RR$} with the line that is rational to the unit circle in $X_g$ because of the transcendental transformation.
In other words, we cannot realize the real condition by treating the real hyperelliptic curves, except for the case where the radius of the circle is infinite, which corresponds to the soliton solutions.

Therefore, we have to treat the hyperelliptic curve $X_g$ with the complex-valued branch points i.e. $b_i \in \CC$.
Then \revs{from \cite{Mat02b}}, $\psi:=$\revs{$\xi$}$/\ii$ formally obeys the focusing \revs{pre-MKdV} equation \revs{
\begin{equation}
\partial_{u_{g-1}}\psi
\revs{+}\frac{1}{2}(\lambda_{2g}+3b_0) \partial_{u_g}\psi
+\frac{1}{8}
\left(
\partial_{u_g} \psi\right)^3
+\frac{1}{4}\partial_{u_g}^3 \psi = 0.
\label{1eq:fprMKdV}
\end{equation}
}
In particular, \cite{Mat02b} shows that $(b_0-x_1)\cdots(b_0-x_g)$ is directly equal to the complex representation of the tangent vector of the loop soliton equation, and also to
$$
\prod_{i=1}^g (b_0-x_i)=
b_0^g-\wp_{g,g}b_0^{g-1}-\wp_{g,g-1}b_0^{g-2}-\cdots-\wp_{g,2}b_0-\wp_{g,1}
\in \mathrm{U}(1),
$$
where $\mathrm{U}(1):=\{z\in \CC\ |\ |z|=1\}$.
This means that its associated hyperelliptic solution $\wp_{gg}$ of the KdV equation must no longer be real valued.

Since $\psi$ is a complex valued meromorphic function on Jacobi variety $J_X$ of $X_g$, i.e., $\psi=\psi_\rr + \ii \psi_\ri$ of the real valued functions  $\psi_\rr$ and $\psi_\ri$, the cubic part generates the term $(\partial_{s} \psi_\ri)^2  \partial_{s} \psi_\rr$, which behaves like a gauge potential.
Thus we encounter a gauged focusing \revs{pre-MKdV} equation from the focusing \revs{pre-MKdV} equation over $\CC$ \revs{(\ref{1eq:fprMKdV}) as}
\begin{equation}
(\partial_t-
A(u)\partial_s)\psi_\rr
           +\frac{1}{8}
\left(\partial_s \psi_\rr\right)^3
+\frac{1}{4}\partial_s^3 \psi_\rr=0.
\label{1eq:gaugedMKdV2}
\end{equation}
It is in contrast to the above case of the defocusing \revs{pre-MKdV} equation (\ref{1eq:dprMKdVxi}), where it is guaranteed that $\partial_{u_g} \revs{\xi_\ri=0}$ simply.
\color{black}
We also remark that if we take $\ii(s,t)$ in $\ii\RR^2$ instead of $(s,t) \in \RR^2$ in (\ref{1eq:gaugedMKdV2}), we also have the real gauged focusing pre-MKdV equation,
\begin{equation}
-(\partial_t-A(u)\partial_s)\psi_\rr
           +\frac{1}{8}
\left(\partial_s \psi_\rr\right)^3
+\frac{1}{4}\partial_s^3 \psi_\rr=0.
\label{1eq:gaugedMKdV2a}
\end{equation}
It means that the requirement of the real solution allows to find $\RR^2$ or $\ii \RR^2$ in the target space of the Abelian integral.

Let $u_a \in \CC^g$ be decomposed as $u_a = u_{a\rr} + \ii u_{a\ri}$, $(a=1, \ldots, g)$.
\cite{MP22} shows that to obtain the real solution of focusing MKdV equation is to find the situation that the following conditions are satisfied for the solutions of (\ref{1eq:gaugedMKdV2}) or (\ref{1eq:gaugedMKdV2a}):
\begin{enumerate}

\item[CI] $\prod_{i=1}^g |x_i - b_a|=$ a constant $(> 0)$,,

\item[CII] $d u_{g\,\ri}=d u_{g-1\, \ri}=0$ or $d u_{g\,\rr}=d u_{g-1\, \rr}=0$, and

\item[CIII] $A(u)$ is a real constant:
if $A(u)=$ constant, (\ref{1eq:gaugedMKdV2}) or (\ref{1eq:gaugedMKdV2a}) is reduced to (\ref{4eq:rMKdV2}), i.e., $\psi_\rr=\phi$. 
\end{enumerate}

\color{black}

However, it is quite difficult to find the real plane $\{(u_g, u_{g-1})\}$ in the Jacobi variety $J_X$ which corresponds to the unit circle valued $(b_0-x_1)\cdots(b_0-x_g)\in \mathrm{U}(1)$.

Therefore, the hyperelliptic solutions of (\ref{4eq:rMKdV2}) have not explicitly and concretely obtained whereas the elliptic function solutions of (\ref{4eq:rMKdV2}) were studied well since Euler's discovery \cite{Euler44}.

As a conclusion of \cite{MP22}, higher genus hyperelliptic curves ($g \ge 3$) might be better to find the solution of (\ref{4eq:rMKdV2}).
Hence this paper is devoted to find the solutions of (\ref{4eq:rMKdV2}) in terms of the meromorphic functions of hyperelliptic curves $X$ over $\CC$ of genus three.

In particular, in this paper, we focus on the real hyperelliptic solutions of the gauged \revs{pre-MKdV} equation (\ref{1eq:gaugedMKdV2}) or (\ref{1eq:gaugedMKdV2a}) of genus three since it is obvious that if $A(u)$ in (\ref{1eq:gaugedMKdV2}) or (\ref{1eq:gaugedMKdV2a}) is constant, then the solution is a solution of (\ref{4eq:rMKdV2}).

\bigskip

\color{black}
We focus on the graph of the Abelian integral and go back to the age of Weierstrass and Baker.

However before considering the higher genus case, we will review how Euler and Lagrange drew the elastica whose tangential angle $\psi$ obeys the static pre-MKdV equation $a \partial_s \psi + (\partial_s\psi)^3/2 + \partial_s^3\psi=0$ for a certain real number $a$ without elliptic function from a modern point of view \cite{Euler44, Lagrange, Mat10, MP22}.

Consider an elliptic curve $E$, $y^2 = (x-e_1)(x-e_2)(x-e_3)$ for certain $e_j$, ($j=1,2,3$) with the projection $\varpi_x: E\to \PP^1$, ($(x,y)\mapsto x$).
For the elliptic curve $E$, we have its Abelian covering $\tE$, whose element $\gamma_{(x,y), \infty}$ is a path from $\infty$ to a point $(x,y)$ in $E$ up to homotopy equivalence; we have a projection $\kappa_E: \tE \to E$, ($\kappa_E(\gamma_{(x,y), \infty})=(x,y)$). Using $\tE$, we define the elliptic integral $\tv : \tE \to \CC$, $\tv(\gamma):= \displaystyle{\int_{\gamma_{(x,y),\infty}} \nu}$, where $\nu=\displaystyle{\frac{dx}{2y}}$.
For each $\gamma \in \tE$, we note that $u := \tv(\gamma)\in \CC$ and $(x,y)=\kappa_X(\gamma)$, we have the differential identity $a' \partial_u \psi + (\partial_u\psi)^3/2 + \partial_u^3\psi=0$ for a certain $a'\in \CC$, and $\psi:=(\log(x-e_1))/\ii$ from the governing equation, $y^2 = \prod (x-e_j)$.

We are interested in a special subspace of the graph of $\tv$, $\cG(\tv):=\{(\gamma, \tv(\gamma))\ |\ \gamma \in \tE\}$ by restricting its domain for drawing the shape of elastica.

Fix an element $\gamma_{(x_0,y_0), \infty}\in \tE$. Since $\gamma_{(x_0,y_0), \infty}$ is a path in $E$, we introduce a subspace $\tS$ in $\tE$ by the homotopy parameter $\ft \in [0,1]$ by $\tS:=\{\gamma_\ft:=\ft \gamma_{(x_0, y_0), \infty} \ | \ \ft \in [0,1]\}$
$\subset \tE$.
The elements in $\tS$ trace $\gamma_{(x_0, y_0), \infty}$.
Then we have the graph of $\cG(\tv|\tS):=\{(\gamma, \tv(\gamma)) \ |\ \gamma \in \tS\}$ $\subset \tE \times \CC$.
Since $\tv(\gamma)$ belongs to $\CC$, we may find $\tS$ such that $\tv(\tS) \subset \RR \subset \CC$ and the induced map $[0,1] \to \RR$, $\tv(\gamma_\ft)$, is a monotone increasing with respect to $\ft$.
We fix such a $\tS$.

By letting $s := \tv(\gamma_\ft)\in \RR$, we find an induced graph structure $\Xi:=\{(s:=\tv(\gamma_\ft),\psi:= (\log((\varpi_x(\kappa_X(\gamma_\ft))-e_1))/\ii\ |\ \ft \in [0,1]\}$, $\{(s, \psi(s)\}=:\cG(\psi)$ whose $\psi(s)$ satisfies the static pre-MKdV equation as a differential identity on $E$.
Further if we find the subspace $\tS$ such that $\psi(s)$ is real for $s\in \RR$ or $x=\varpi_x(\kappa_X(\gamma_\ft))$ belongs to the unit circle whose center $e_1$, we have a real elliptic solution of the pre-MKdV equation.
To do so, Lagrange introduced the angle expression, $y^2 = 4\ee^{\ii \psi}(1 - k^2 \sin^2 (\psi/2))/k^2$ in \cite{Lagrange} for a certain $k$ \cite{MP22}.
By integrating $\ee^{\ii\psi(s)} ds$ for such a situation, we can draw Euler's elastica without an elliptic function, as Euler and Lagrange essentially did.

Although we will directly parameterize the real subspace $\tS$ in the target space of the Abelian integral by integrating a differential equation from first, we extend this idea to the case of hyperelliptic curves.

\bigskip

Thus we go back to the age of Weierstrass and Baker to look at hyperelliptic curves $X$, their Abelian maps and meromorphic functions from a modern point of view.
We have the Abelian covering $\tX$ of $X$ by abelianization of the path space of $X$, $\kappa_X : \tX \to X$ ($\gamma_{P, \infty} \mapsto P$), and the projection $\varpi_x: X \to \PP^1$, ($(x,y)\mapsto x$).
By the Abelian integral $\tv : S^3 \tX \to \CC^3$, for a certain point $\gamma=(\gamma_{P_1,\infty}, \gamma_{P_2,\infty}, \gamma_{P_3,\infty}) \in S^3 \tX$, we have a point $\tv(\gamma_{P_1,\infty}, \gamma_{P_2,\infty}, \gamma_{P_3,\infty}) \in \CC^3$.
The Jacobian $J_X$ is defined as $\kappa_J : \CC^3 \to J_X:=\CC^3 / \Gamma_X$.
Due to the Abel-Jacobi theorem \cite{FarkasKra}, for a point $u \in \CC^3$, we may find a point $\gamma$ in $S^3 X$ such that $\tv(\gamma)=u$.
The above $\psi$ in (\ref{1eq:fprMKdV}) is a meromorphic function of $S^3 X$, which satisfies the gauged focusing pre-MKdV equation (\ref{1eq:gaugedMKdV2}) or (\ref{1eq:gaugedMKdV2a}) as an differential identity on $S^3 X$.

If we have a subspace $\tfS \subset S^3 \tX$ whose image of $\kappa_X$ is related to the real value of $\psi$ and whose image of $\tv$ belongs to the real vector space $\fL:=\displaystyle{\left\{\begin{pmatrix}
0 \\ u_{2\, \rr} \\u_{3\, \rr}\end{pmatrix} \Bigr|\ u_{a\,\rr}\in \RR\right\}}$ $=\RR^2$ (or $\ii \RR^2$) of $\CC^3$, then we have the real solution of (\ref{1eq:gaugedMKdV2}) (or (\ref{1eq:gaugedMKdV2a})).

More precisely speaking, we are looking for a subspace $\fX:=\{(\gamma, \tv(\gamma))\ | \gamma \in \tfS \subset S^3\tX, \tv(\gamma) \in \fL\}$ so that $\tfS$ corresponds to real $\psi$ in the graph of $\tv$, $\cG(\tv)=\{ (\gamma, \tv(\gamma)) \ |\ \gamma \in S^3 \tX\} \subset S^3 \tX \times \CC^3$ as fibers of suitable points in the moduli space of the hyperelliptic curves.
Then it is obvious that the conditions CI and CII are satisfied.
Thus the subspace $\fX$ can be regarded as the graph $\cG(\psi|\fL)=\{(u, \psi(u)) \ |\ u \in \fL\}$.

In this paper, we cannot find such a subspace $\fX$ but we find a subspace of $\cG(\tv)$ such that both pre-image and image of $\tv$ satisfy the condition CI and CII.
Let us show the strategy as follows.

Instead of $\tfS$, we consider $\tSS \subset S^3 \tX$.
We assume that $\varpi_x\kappa_X\tSS$ consists of the symmetric product of three points $x_i$ in the arc of the unit circle at a branch point $(b_0, 0)$ such that $-\ii\log(x_i-b)$ is real.
Our purpose is to find the proper set of the real subspace $\tSS$ in $S^3 \tX$ and the basis of $\RR^2$ in the target space $\CC^3$ of $\tv$.
Since for the case of the genus one, the angle expression of the elliptic curves plays an important role as Lagrange introduced, we also introduce the angle expression of the hyperelliptic curves such that $\psi$ may be real value in $S^3 X$ \cite{Mat07, MP22} to consider the behavior of the above arc.

As the angle expression (of the Jacobi elliptic function) is related to the double covering of the elliptic curve $E$ of the Weierstrass standard form, the angle expression in the hyperelliptic integral $\tv$ is also connected with the double covering $\hX$ of $X$ \cite{Mat07} which was introduced by Weierstrass \cite{Wei54}.
Here $\tX$ and $\hX$ are set so that both are consistent as in (\ref{2eq:al_hyp_kappas}). 
By using these tools, we have attacked the problem to find the real solution of the focusing MKdV equation \cite{Mat07, MP22}.

We are concerned with the graph of the Abelian integral $\tv$ such that the elements satisfy the conditions CI and CII.
The target space of $\tv$ is flat but the structure of its preimage $S^3 \tX$ is complicated.
It is a problem how to find certain real subspace $\tSS=\{\gamma = (\gamma_1, \gamma_2, \gamma_3)\} \subset S^3 \tX$.

We follow the strategy by Weierstrass and Baker. 
Weierstrass and Baker investigated the meromorphic functions on $S^3 X$ and $J_X$ by using the pullback between these cotangent spaces, i.e., $\tv^*: T^* \CC^3 \to T^* S^3 \tX$.
The Abel-Jacobi theorem shows that it is isomorphic except for some singular locus.
In other words, $\tv_*$ which is given by a transition matrix is invertible for regular places.
They used the the matrix representation of $\tv_*$ to study the properties of the sigma functions and found the sine-Gordon equation implicitly \cite{Wei54} and the KdV hierarchy explicitly \cite{Baker03}.

First, we assume non-singular locus for the birational map $\tv$.
We handle the pullbacks 
 $\tv^*: T^* \CC^3 \to T^* S^3 \tX$ and $\tv^{-1*}: T^* S^3 \tX\to T^* \CC^3$.

In order to find the condition that the image of $\tv$ is a real vector space,
we deal with $L_0:=\left\{\be_1s = \begin{pmatrix} s/2 \\ -s/2\\ s \end{pmatrix}\ \Bigr|\ s\in \RR\right\}\subset$  $\LL:=\left\{\be_1s+\be_2 t=\begin{pmatrix} t+s/2 \\ -s/2\\ s \end{pmatrix}\ \Bigr|\ s,t\in \RR\right\}$ as a subspace of the target space
$\CC^3=\displaystyle{\left\{\begin{pmatrix} u_1\\ u_2\\ u_3\end{pmatrix} \Bigr|\ u_i \in \CC\right\}}$ of the Abelian integral $\tv$.
Here let $\displaystyle{\be_1:=\begin{pmatrix} 1/2 \\ -1/2\\ 1 \end{pmatrix}}$
and  $\displaystyle{\be_2:=\begin{pmatrix} 1 \\ 0\\ 0 \end{pmatrix}}$.
As in Proposition \ref{4pr:reality_g3}, for the basis of the cotangent space $\be_1 ds$ and $\be_2 ds$ in $T^* \LL$, we find that $\tv^* (\be_1 ds)=\cV_1 ds$ and $\tv^*(\be_2 dt)=\cV_2 dt$ are real valued one-forms in $T^* \tSS$ for a certain $\tSS$ so that $\varpi_x \kappa_X\tSS$ consists of the arcs of the unit circle and a certain $X$.
The condition for $X$ means that the branch points except $b_0$ of $X$ belong to the unit circle $|x-b_0|=1$ as in Fig.~\ref{fg:Fig01}, which is the same as the condition that $y$ in (\ref{4eq:HEcurve_phi}) must be related to the real value in a certain sense.
By using the angle expression, we consider the condition.

For $\gamma\in \tSS$, and $\varpi_x \kappa_X\gamma = (x_{1}, x_{2}, x_{3})$.
the condition CI means that $\varphi_{i} := (\log (x_{i} - b_0))/2\ii$ is real.

For simplicity, let $t=0$ and focus on $L_0 \subset \LL$.
In the target space $\CC^3$ of $\tv$, we naturally have $\int^s_0 \be_1 ds = \be_1 s$ for $s \in \RR$.
We interpret $d \gamma = (d \varphi_1, d\varphi_2, d\varphi_3) =\tv^*(\be_1 ds)=\cV_1 ds$  (\ref{eq:g3CIII}) as a differential equation.
Consider it as an orbit by integrating $\tv^*(\be_1 ds)=\cV_1 ds$, $\gamma_s=\displaystyle{\int^s\cV_1 ds}$ for an initial point $\gamma_0 \in S^3 \tX$ for $s\in \RR$.
Proposition \ref{4pr:reality_g3} shows that the integration of $d \gamma=\cV_1 ds$ (\ref{eq:g3CIII}) provides $\tS_0$.
Due to the branch points, it turns out that the orbit $\gamma_s$ in $S^3 \tX$ consists of loops among branch points of $X$ or $\hX$. 
We will set $\cV_1 ds$ so that $\varphi_{s,i}$ of $\gamma_s$ must be real or $x_{s,i}$ of $\gamma_s$ must belong to the unit circle for any $s\in \RR$.
Further, it is obvious that $\tv(\gamma_s)=\be_1 s$ or $\tv^{-1}(\be_1 s) = \gamma_s$ from the assumption.

From the construction, the graph of $\tv^{-1}$ for the obtained $L_0$ is $\cG(\tv^{-1}|L_0)=\{(\be_1 s, \tv^{-1}(\be_1 s) \ |\ \be_1 s\in L_0\} \subset L_0 \times \tS_0$.
It means that we find a graph structure in which the conditions CI and CII are satisfied.

Then the meromorphic functions at a point $\gamma$ in $S^3 \tX$ satisfies (\ref{1eq:gaugedMKdV2})  or (\ref{1eq:gaugedMKdV2a}) as an differential identity.
We also estimate $\partial_{u_{3\ri}} \psi$ at $\gamma$ by using the pullback $\tv^*$ and the data of the point $\gamma$ locally as in (\ref{eq:g3CIIIi}).
If the gauged field $A$ in (\ref{1eq:gaugedMKdV2}) or (\ref{1eq:gaugedMKdV2a}) is constant or satisfies the condition CIII, we might have a real solution of the focusing MKdV equation. 

Precisely speaking, as mentioned above, we cannot find a real solution the solution of the gauged focusing pre-MKdV equation (\ref{1eq:gaugedMKdV2a}) in this paper because $L_t$ differs from $\displaystyle{\RR\begin{pmatrix}0\\0\\1\end{pmatrix}}$, but we will find real cross-sections along $L_0$ of real solutions of (\ref{1eq:gaugedMKdV2a}) as in Theorem \ref{th:solgMKdV_R}.

\color{black}
Further, as the Abel-Jacobi map is a bi-rational map from $S^3 \hX \to \hJ_X$, there is a singular locus in $\tS$ and $L$. 
We have to check the existence of $\tS$ as a locus of the symmetric product of $\tX$ as in Section 5.
For a point $\gamma=(\gamma_1, \gamma_2, \gamma_3) \in S^3 \tS$, if their neighborhoods $U_{\gamma_i}$ at $\gamma_i$ are disjoint, we can discriminate them because $d \gamma$ is essentially defined as $d \gamma_i$'s.
However the point $\gamma_i = \gamma_j$ ($i \neq j$) in $S^3 \tX$, we have to deal with them carefully as in Section 5.
Lemma \ref{4lm:intersection} shows that we also define the orbit even for such a case.

\color{black}

\revs{
In other words, we find} Proposition \ref{4pr:reality_g3}, which guarantees the reality conditions of both $S^3X$ and the Jacobian $J_X$.
The relation in Proposition \ref{4pr:reality_g3} shows a novel real part of the hyperelliptic solution of the gauged MKdV equation (\ref{1eq:gaugedMKdV2}) or  (\ref{1eq:gaugedMKdV2a})  as in Theorem \ref{th:solgMKdV_R}.
Theorem \ref{4th:reality_g3} provides that when the gauge field $A$ is constant, it is reduced to the solution of the MKdV equation.

The content is following:
Section 2 reviews the hyperelliptic solutions of the focusing \revs{pre-MKdV} equation over $\CC$ \revs{of genus three} in Theorem \ref{4th:MKdVloop} following \cite{Mat02b,MP15,MP22}.
Section 3 is devoted to the \revs{angle expression} of the hyperelliptic curves of genus three.
Section 4 provides local properties of the solutions of the gauged \revs{pre-MKdV} equation (\ref{1eq:gaugedMKdV2}) as in Theorem \ref{pr:solgMKdV}.
Based on this, we investigate the global behavior of the hyperelliptic solutions of genus three of the gauged \revs{pre-MKdV} equation in Section 5.
\revs{For the hyperelliptic solution, we will find orbits in $S^3 \hX$ as a preimage of the Abel-Jacobi map $v: S^3 \hX \to \hJ_X$ for a real line $\kappa_J L$ in the Jacobian $\hJ_X$ of $X$.}
In particular, the behavior at the intersection of two orbits and at branch points is discussed. 
We conclude that the solutions are well-defined globally in Theorem \ref{th:solgMKdV_R}.
As Discussion in Section 6, we describe the real solution of the focusing MKdV equation (\ref{4eq:rMKdV2}) as in Theorem \ref{4th:reality_g3}.
There we also demonstrate the numerical evaluation of the gauge potential.
Section 7 gives the conclusion of this paper.

\section{Hyperelliptic solutions of focusing \revs{pre-MKdV} equation$/\CC$ \revs{of genus three}}
\label{sec:HESGE}

\color{black}
\cite{MP22} concludes that in order to obtain the solution of (\ref{4eq:rMKdV2}) based on hyperelliptic function theory, we should handle hyperelliptic curves  of genus $g>2$. 
In this paper, we consider a hyperelliptic curve $X$ of genus three over $\CC$,
\begin{equation}
X=\left\{(x,y) \in \CC^2 \ |
\ y^2 + (x-b_0)(x-b_1)(x-b_2)\cdots(x-b_{6})=0\right\}
\cup \{\infty\},
\label{4eq:hypC}
\end{equation}
where $b_i$'s are mutually distinct complex numbers.
Let $\lambda_{6}=\displaystyle{-\sum_{i=0}^{6} b_i}$ and $S^k X$ be the $k$-th symmetric product of the curve $X$. 
Further for the Abelian integral, we introduce the Abelian covering $\tX:=\Gamma_\infty X$ of $X$ by abelianization of the path-space of $X$, $\kappa_X: \tX \to X$, ($\gamma_{P, \infty} \mapsto P$) \cite{M24}.
Here $\gamma_{P, \infty}$ means a path from $\infty$ to $P$.
$S^k \tX$ also means the  the $k$-th symmetric product of the space $\tX$. 
The Abelian integral $\tv : S^3  \tX \to \CC^3$ is defined by its $i$-th component $\tv_i$ $(i =1,2,3)$,
\begin{equation}
\tv_i(\gamma_1,\gamma_2,\gamma_3)=\sum_{j=1}^3
 \tv_i(\gamma_j), \quad
\tv_i(\gamma_{(x,y), \infty}) = \int_{\gamma_{(x,y), \infty}} \nuI{i},\quad
\nuI{i} = \frac{x^{i-1}d x}{2y}.
\label{4eq:firstdiff}
\end{equation}
Then we have the Jacobian $J_X$ : $\kappa_J: \CC^3 \to J_X=\CC^3/\Gamma_X$, where $\Gamma_X$ is the lattice generated by the period matrix for the standard homology basis of $X$.
Due to the Abel-Jacobi theorem \cite{FarkasKra}, we also have the bi-rational map $v$ from $S^3 X \to J_X$ by letting $v:=\tv$ modulo $\Gamma_X$.
We refer to $v$ as the Abel-Jacobi map.

\bigskip
Further, we implicitly consider the al function $\al_a(u) = \sqrt{(b_a-x_1)(b_a-x_2)(b_a-x_3}$ due to the ramification degree, two, at the branch point $B_a:=(b_a, 0) \in X$ ($\varpi_x: X\to \revs{\PP^1}$), which is originally defined by Weierstrass 1854 \cite{Wei54, Baker98}.
Fix $a=0$.
The square root leads the transformation of $w^2 = (x-b_0)$, i.e., the double covering $\hX$ of the curve $X$, $\hvarpi: \hX \to X$, although the precise arguments are left to the Appendix in \cite{MP15}.
Since $\hX$ is regarded as a path space of $X$ as in \cite{Baker98} and $\tX$ is a covering of $\hX$, we have a natural commutative diagram,
\begin{equation}
\xymatrix{ 
 \tX \ar[dr]^{\kappa_X}\ar[r]^-{\kappa_\hX}& \hX \ar[d]^-{\varpi_\hX} \\
  & X,
}\label{2eq:al_hyp_kappas}
\end{equation}
This transformation is a natural generalization of the relation between Weierstrass $\wp$ function and Jacobi $\sn, \cn, \dn$ functions because the Jacobi function consists of $\sqrt{\wp-e_i}$, ($i=1,2,3$) of genus one.

The curve $\hX$ is given by $z^2 + (w^2-e_1)\cdots(w^2 - e_{6})=0$, where $z:=y/w$, and $e_j := b_j- b_0$, $j=1, \ldots, 6$.
Here $z$ is introduced due to the normalization (blow-up) for the singularity at $w=0$.
Since the genus of $\hX$ is five, we have five holomorphic one-forms,
$$
\hnu_j:= \frac{w^{j} d w}{z}, \quad (j=1, 2, 3, \ldots, 5).
$$
and the Jacobi variety, $J_{\hX}$ of $\hX$ is given by the complex torus $J_{\hX}=\CC^5 /\Gamma_{\hX}$ for the lattice $\Gamma_{\hX}$ given by the period matrix.
As in \cite[Appendix, Proposition 11.9]{MP15}, we have the correspondence $\hvarpi^*\nuI{i}=\hnu_{2i-2}$, $(i=1,2,3)$ and thus the Jacobian $J_{\hX}$  contains a subvariety $\hJ_X\subset J_{\hX}$ which is a double covering of the Jacobian $J_X$ of $X$, $\hvarpi_J: \hJ_X \to J_X$, and $\hkappa_J : \CC^3 \to \hJ_X:= \CC^3/(\Gamma_\hX\cap \CC^3)$.

Since for each branch point $B_j:=(b_j, 0)\in X$, we have double branch points $\hB^\pm_j:=(\pm\sqrt{e_j},0) \in \hX$ as illustrated in Fig.~\ref{fg:Fig00}.

Similar to the Jacobi elliptic functions, $\hJ_X=\CC^3/\hGamma_X$ is determined by the same Abelian integral $\tv$, and thus we use the same symbol $\tv$ as $\tv : S^3\tX \to \CC^3$ for $\hX$ \cite{MP15}.

\begin{figure}
\begin{center}

\includegraphics[width=0.6\hsize, bb= 0 0 580 487]{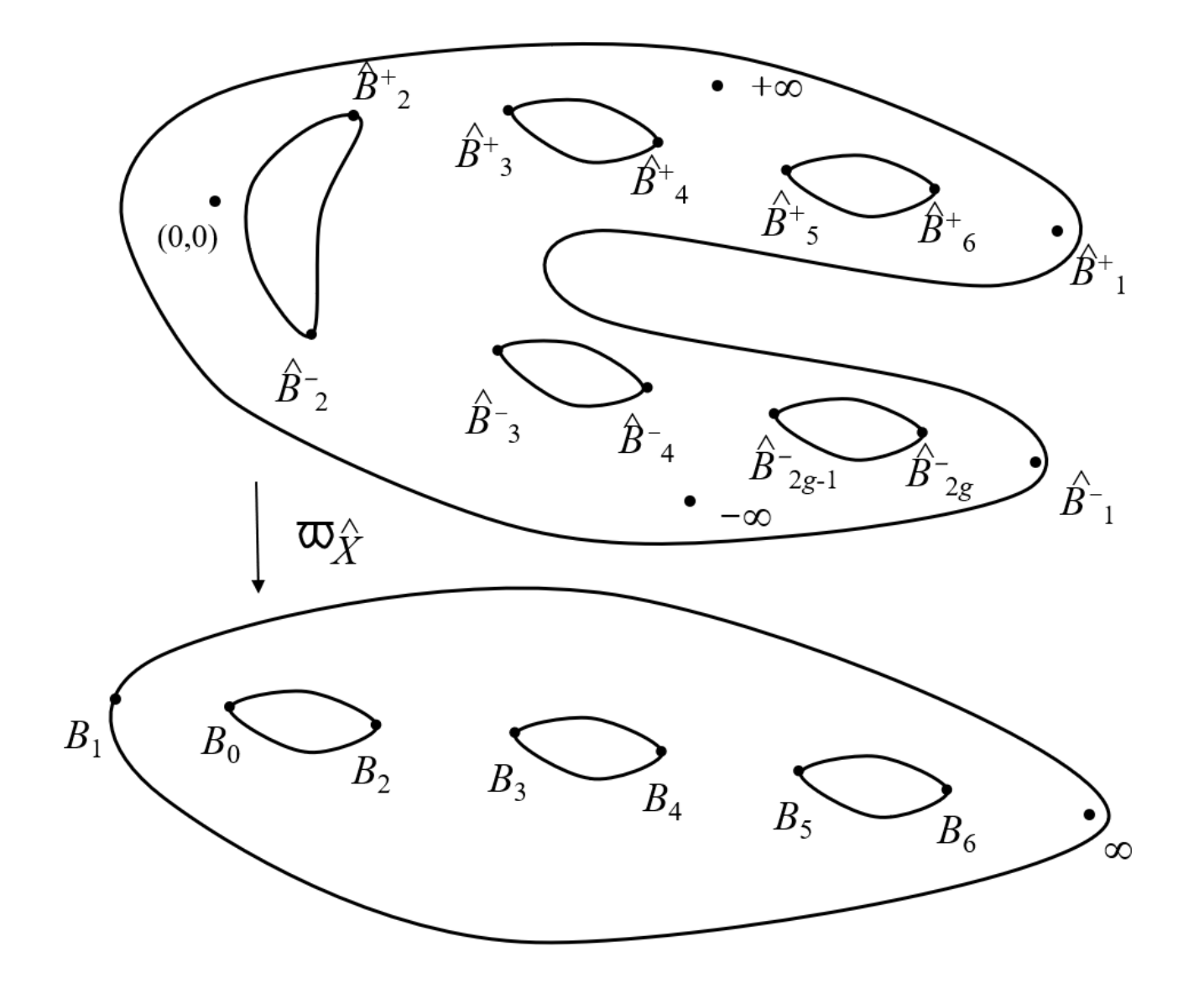}

\end{center}

\caption{The double covering $\hvarpi: \hX_g \to X_g$,
$\hvarpi: (w, z) \mapsto (w^2+b_0, zw)=(x,y)$.
}\label{fg:Fig00}
\end{figure}

\color{black}

\cite{Mat02b} shows the hyperelliptic solutions of the MKdV equation over $\CC$.

\begin{theorem} {\textrm{\cite{Mat02b}}}
\label{4th:MKdVloop}
For $((x_1,y_1),\cdots, (x_3,y_3)) \in S^3 X$, the fixed branch point $(b_0, 0)$, and $u:= v( (x_1,y_1),$ $\cdots,(x_3,y_3))$,
$$
\displaystyle{
   \psi(u) :=
-\ii \log (b_0-x_1)(b_0-x_2)(b_0-x_3)
}
$$
satisfies the pre-MKdV equation over $\CC$,
\begin{equation}
\revs{(}\partial_{u_{2}}-\frac{1}{2}
(\lambda_{6}+3b_0)
          \partial_{u_{3}})\psi
           +\frac{1}{8}
\left(\partial_{u_3} \psi\right)^3
 +\frac{1}{4}\partial_{u_3}^3 \psi=0,
\label{4eq:loopMKdV2}
\end{equation}
where $\partial_{u_i}:= \partial/\partial u_i$ as an differential identity in $S^3 X$ and $\CC^3$.
\end{theorem}

We, here, emphasize the difference between the focusing \revs{pre-MKdV} equations (\ref{4eq:rMKdV2}) over $\RR$ and (\ref{4eq:loopMKdV2}) over $\CC$.
In (\ref{4eq:rMKdV2}), $\phi$ is a real valued function over $\RR^2$ but $\psi$ in (\ref{4eq:loopMKdV2}) is a complex valued function over \revs{$\CC^2 \subset \CC^3$}.
The difference is crucial since we want to obtain solutions of (\ref{4eq:rMKdV2}), not (\ref{4eq:loopMKdV2}).
However, the latter is expressed well in terms of the hyperelliptic function theory.
Although it goes beyond this paper, our ultimate goal is to construct the solutions of (\ref{4eq:rMKdV2}) based on the solutions of (\ref{4eq:loopMKdV2}).

As mentioned in \cite[(11)]{MP22}, we describe the difference.
By introducing real and imaginary parts, $ u_a = u_{a\,\rr} + \ii u_{a\,\ri}$, $(a=1,2,3)$, and $ \psi = \psi_{\rr} + \ii \psi_{\ri}$, the real part of (\ref{4eq:loopMKdV2}) is reduced to the gauged \revs{pre-MKdV} equation with gauge field $A(u)=(\lambda_{2g}+3b_a-\frac{3}{4}(\partial_{u_{g}\, \rr}\psi_\ri)^2)/2$,
\begin{equation}
\revs{(}\partial_{u_{2}\, \rr}-
A(u)\partial_{u_{3}\, \rr})\psi_\rr
           +\frac{1}{8}
\left(\partial_{u_3\, \rr} \psi_\rr\right)^3
+\frac{1}{4}\partial_{u_3\, \rr}^3 \psi_\rr=0
\label{4eq:gaugedMKdV2}
\end{equation}
by the Cauchy-Riemann relations as mentioned in \cite[(11)]{MP22}.

\bigskip
In order to obtain a solution of (\ref{4eq:rMKdV2}) in terms of the data in Theorem \ref{4th:MKdVloop},  the following conditions must be satisfied \cite{MP22}:

\begin{condition}\label{cnd}
{\rm{
\begin{enumerate}

\item[CI] $\prod_{i=1}^3 |x_i - b_a|=$ a constant $(> 0)$ in Theorem \ref{4th:MKdVloop},

\item[CII] $d u_{3\,\ri}=d u_{2\, \ri}=0$ \revs{or $d u_{3\, \rr}=d u_{2\, \rr}=0$} in Theorem \ref{4th:MKdVloop}, and

\item[CIII] $A(u)$ is a real constant:
if $A(u)=$ constant, (\ref{4eq:gaugedMKdV2}) is reduced to (\ref{4eq:rMKdV2}), i.e., $\psi_\rr=\phi$. 
\end{enumerate}
}}
\end{condition}
\revs{
It is obvious that if we have the solutions $\psi_\rr$ of (\ref{4eq:gaugedMKdV2}) satisfying the conditions CI-CIII, $\partial_{u_3, \rr}\psi_{\rr}/2$ obeys the focusing MKdV equation (\ref{eq:MKdV1}).}

However, in this paper we focus on the conditions CI and CII and the real hyperelliptic solutions of the gauged \revs{pre-MKdV} equation (\ref{4eq:gaugedMKdV2}) of genus three as a step to the solution of the real focusing \revs{pre-MKdV} equation (\ref{4eq:rMKdV2}).

\section{Hyperelliptic curves of genus three in angle expression}\label{sec:g=3}
\revs{
In this paper, we mainly investigate the conditions CI and CII in Condition \ref{cnd}, and thus we introduce the analge expression \cite{Mat07, MP22} as mentioned in Introduction.
}

We restrict the moduli (rather, parameter) space of the curve $X$ by the following.
We choose coordinates $u = {}^t(u_1, u_2, u_3)$
 in  $\CC^3$;
$u_i = u_i^{(1)}+u_i^{(2)}+u_i^{(3)}$, where $u_i^{(j)}
=v_i((x_j, y_j))$ for $(x_j, y_j) \in X$.
There are the projection $\varpi_x : X \to \revs{\PP^1}$, $((x,y) \mapsto x)$, and similarly $\hvarpi_x : \hX \to \revs{\PP^1}$.

We let $b_0=-1$ and $e_j := b_j - b_0$ $(j=1,2,\ldots,6)$ satisfying the following relations,
$$\sqrt{e_{2a-1}} = \alpha_a +\ii \beta_a,
\quad
\sqrt{e_{2a}} = \alpha_a -\ii \beta_a,
$$
where $\alpha_a, \beta_a\in \RR$, $a,b =1,2,3$, satisfying 
$\alpha_a^2 + \beta_a^2 = 1$.

We recall $w^2 = (x-b_0)$ and $w=\ee^{\ii \psi}$.
For a {\lq}real{\rq} expression of (\ref{4eq:hypC}), we use the following transformation, which is a generalization of {\lq}the angle expression{\rq} of the elliptic integral as mentioned in \cite{MP22}.
\begin{lemma}
$(w^2-e_1)(w^2-e_2)=\revs{4}\frac{1}{k_1^2}\ee^{2\ii \varphi}
(1-k^2\sin^2\varphi)$, where 
$$
w= \ee^{\ii \varphi}, \quad 
k_1 = \frac{2\ii\sqrt[4]{e_{1}e_{2}}}{\sqrt{e_{1}}- \sqrt{e_{2}}}
=\frac{1}{\beta_a}, \quad
 e_1 e_2 = 1.
$$
\end{lemma}

\begin{proof}
Let $ e_1 e_2=1$.
We recall the double angle formula $\cos 2\varphi = 1-2\sin^2 \varphi$.
\begin{gather*}
\begin{split}
(w^2 - e_1)(w^2 -e_2)& =w^2(w^2-(e_1+e_2) +e_1 e_2 w^{-2})\\
&=w^2 \left(\ee^{2\ii\varphi}+\ee^{-2\ii\varphi}
-\frac{e_1+e_2}{1}\right)\\
&=2w^2 \left(\cos(2\varphi) - \frac{e_1+e_2}{2}\right)\\
&=-w^2 
\left(
e_1+e_2-2\sqrt{e_1e_2}+4\sin^2 \varphi\right) \\
&=\revs{4}w^2 \frac{1}{k_1^2}
\left(1-k_1^2\sin^2 \varphi\right), \\
\end{split}
\end{gather*}
where $(e_1+e_2-2\sqrt{e_1e_2})=(\sqrt{e_1}-\sqrt{e_2})^2 = e_1^{-1}(e_1+1)^2=-4/k_1^2$.
\qed
\end{proof}

Under these assumptions, we have the real extension of the hyperelliptic curve $X$ by $(\ee^{\ii\varphi}, y/\ee^{\ii\varphi})\in \hX$. 
The direct computation shows the following:

\begin{lemma} \label{4lm:g3gene_y2}
Let $ \ee^{2\ii\varphi} :=(x-b_0)$,
(\ref{4eq:hypC}) is written by
\begin{gather}
y^2=-64 \frac{4\ee^{8\ii\varphi}}{k_1^2 k_2^2k_3^2} 
(1-k_1^2 \sin^2 \varphi)(1-k_2^2 \sin^2 \varphi)
(1-k_3^2 \sin^2 \varphi),
\label{4eq:HEcurve_phi}
\end{gather}
where 
$\displaystyle{
k_a = \frac{2\ii\sqrt[4]{e_{2a-1}e_{2a}}}{\sqrt{e_{2a-1}}- \sqrt{e_{2a}}}
=\frac{1}{\beta_a}}$, $(a=1,2,3)$.
\end{lemma}

As we are concerned with the situation that $y/\ee^{4\ii \varphi}$ is real or pure imaginary, we assume that the branch points surround the circle whose center is $(b_0,0)$ and radius is $1$.
Noting that we handle the double covering $\hX$ with twelve branch points, we define $\varphi_{\fb a}^{+\pm}:=\pm \sin^{-1}(1/k_a)$ and $\varphi_{\fb a}^{-\pm}
=\pi - \varphi_{\fb a}^{+\pm}$ $(a= 1, 2, 3)$ as in Figure \ref{fg:Fig01}.
\revs{
In other words, we have twelve branch points $(\ee^{\varphi_{\fb a}^{\pm\pm}}, 0)$ in $\hX$, $a=1, 2, 3$.
}

\begin{assumption}\label{Asmp}
{\rm{
For the case Figure \ref{fg:Fig01} (a) $k_1> k_2 > k_3>1.0$, we assume $\varphi_\fb^\pm:=\varphi_{\fb1}^{+\pm}$ whereas for the case Figure \ref{fg:Fig01} (b) we let $\varphi_\fb^\pm:=\varphi_{\fb1}^{\pm+}$ for $(k_3> k_2 > k_1 > 1.0)$.
}}
\end{assumption}

\bigskip

\begin{figure}
\begin{center}

\includegraphics[width=0.42\hsize, bb= 0 0 641 668]{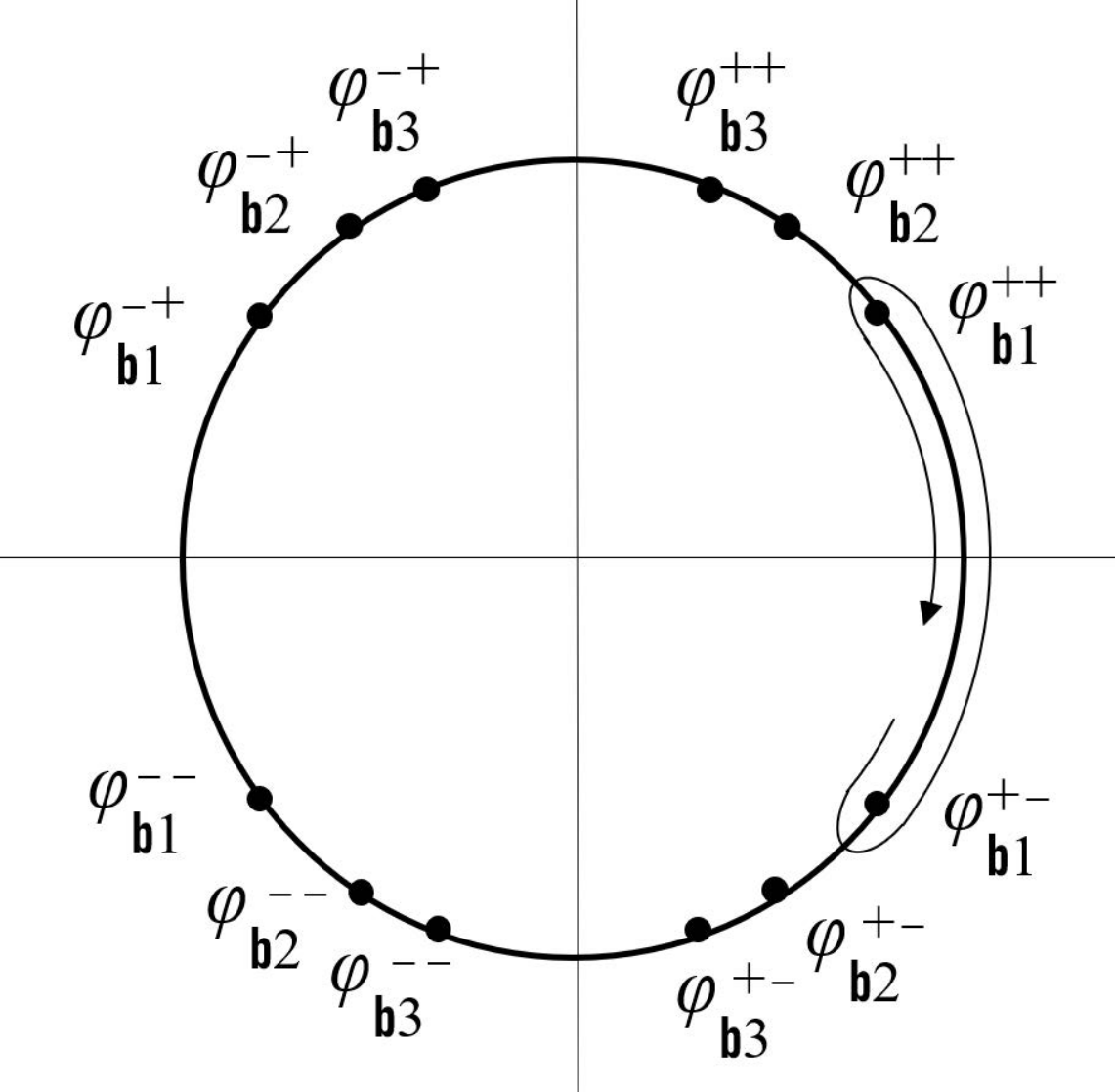}
\hskip 0.1\hsize
\includegraphics[width=0.39\hsize, bb= 0 0 539 436]{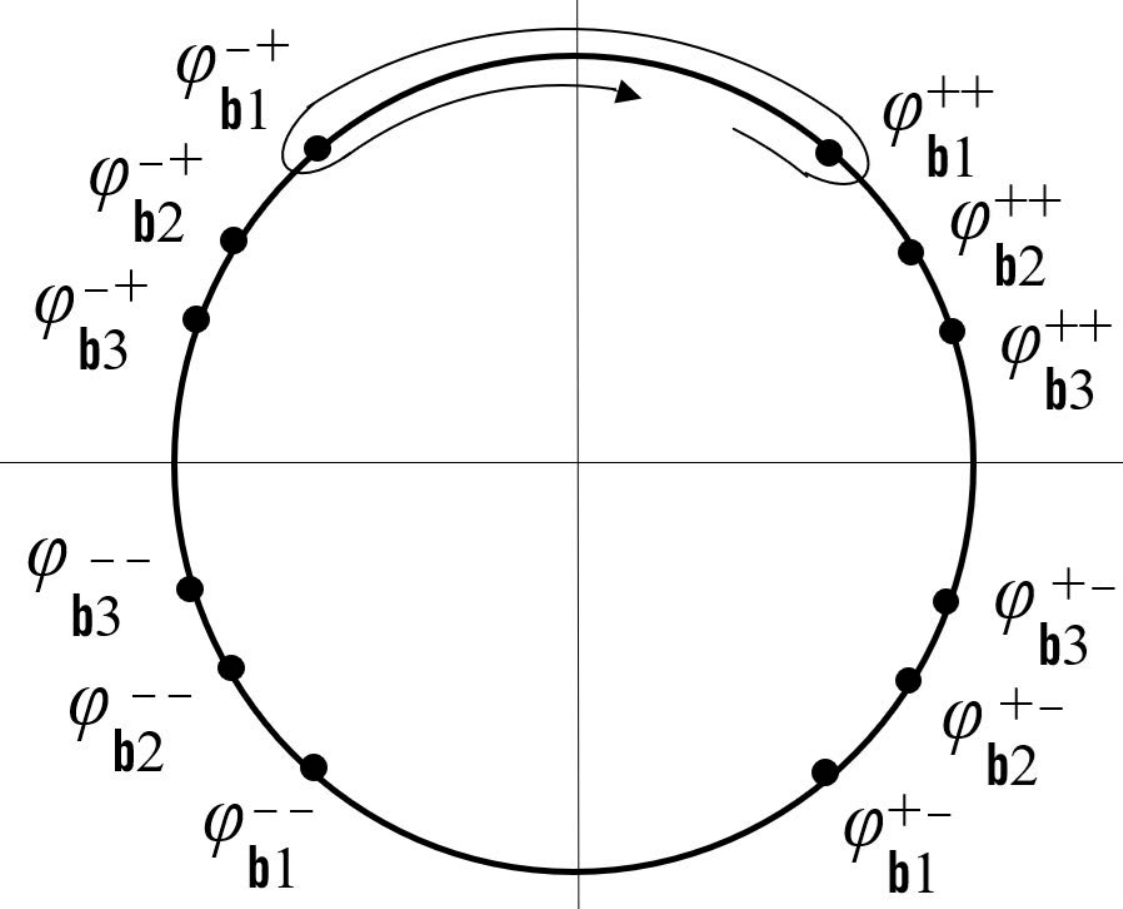}

(a) \hskip 0.35\hsize (b)

\end{center}

\caption{
The orbits of each $\varphi_i$ in the quadrature:
(a): $k_1> k_2 > k_3>1.0$.
(b): $k_3>k_2>k_1>1.0$.
}\label{fg:Fig01}
\end{figure}

\bigskip

We consider a point $((x_1, y_1),\ldots,(x_3,y_3))$ in $S^3 X$ under the condition CI, $|x_j-b_0|=1$, ($b_0 = -1$).
We define the variable $\varphi_j$ by $x_j=  \ee^{\ii \varphi_j}
(\ee^{\ii \varphi_j}+ b_0 \ee^{-\ii \varphi_j})$ $=2 \ii  \ee^{\ii \varphi_j} \sin \varphi_j$, $(j=1,2,3)$.
Noting
$d x_j = 2 \ii  \ee^{2\ii \varphi_j}d\varphi_j$ and
 $x_j^\ell d x_j = (2\ii)^{\ell+1}  \ee^{(2+\ell)\ii \varphi_j}\sin^\ell \varphi_j\ d \varphi_j$,
we have the holomorphic one forms
$(\nuI{1}^{(j)}, \nuI{2}^{(j)},  \nuI{3}^{(j)})$ $(j=1,2,3)$, 
\begin{equation}
\left(
\frac{ \ee^{-2\ii\varphi_j}\ d \varphi_j}{
8  K(\varphi_j)},
\frac{-\ii \ee^{-\ii\varphi_j}\sin(\varphi_j)\ d \varphi_j}{
4  K(\varphi_j)},
\frac{ -\sin^2 \varphi_j\ d \varphi_j}{
2K(\varphi_j)}\right),
\label{eq:3.3_1stdiff}
\end{equation}
where $K(\varphi):=\tgamma\tK(\varphi)$, $\displaystyle{
\tK(\varphi):=\frac{\sqrt{
 (1-k_1^2 \sin^2 \varphi)(1-k_2^2 \sin^2 \varphi)
(1-k_3^2 \sin^2 \varphi)}}
{k_1k_2k_3}}$, and $\tgamma = \pm 1$.
$\tgamma$ represents the sheet from $\hvarpi_x:\hX \to \revs{\PP^1}$, 
$(( \ee^{\ii \varphi},
\ii 8\tgamma \tK(\varphi)\ee^{3\ii \varphi}) \mapsto \ee^{\ii \varphi})$.

Using the ambiguity $\tgamma$, we handle $-(\nuI{1}^{(j)}, \nuI{2}^{(j)},  \nuI{3}^{(j)})$ rather than $(\nuI{1}^{(j)}, \nuI{2}^{(j)},  \nuI{3}^{(j)})$ $(j=1,2,3)$ from here.
We note that $( \ee^{\ii \varphi}, \ii8 K\ee^{3\ii\varphi})$ denotes a point in $\hX$.

Then we obviously have the following lemmas:

\begin{lemma} \label{4lm:dudphi}
For $( \ee^{\ii \varphi_j}, \ii8 K_j\ee^{3\ii\varphi_j})_{j=1, 2, 3}\in S^3\hX$, where $K_j:=\tgamma_j\tK(\varphi_j)$, ($j=1, 2, 3$), the following holds:

\noindent
$$\displaystyle{
\begin{pmatrix} d u_1 \\ d u_2\\ du_3\end{pmatrix}
=-
\begin{pmatrix}
\frac{ \ee^{-2\ii\varphi_1}}{8 K_1}&
\frac{ \ee^{-2\ii\varphi_2}}{8 K_2}&
\frac{ \ee^{-2\ii\varphi_3}}{8 K_3}\\
\frac{\ii \ee^{-\ii\varphi_1}\sin(\varphi_1)}{4 K_1}&
\frac{\ii \ee^{-\ii\varphi_2}\sin(\varphi_2)}{4 K_2}&
\frac{\ii \ee^{-\ii\varphi_3}\sin(\varphi_3)}{4 K_3}\\
\frac{- \sin^2(\varphi_1)}{2 K_1}&
\frac{- \sin^2(\varphi_2)}{2 K_2}&
\frac{- \sin^2(\varphi_3)}{2 K_3}\\
\end{pmatrix}
\begin{pmatrix} d \varphi_1 \\ d \varphi_2 \\d \varphi_3\end{pmatrix}
}.
$$
Let the matrix be denoted by $\cL$.
Then the determinant of $\cL$, \revs{
$$
\det(\cL)=\displaystyle{
\frac{\sin(\varphi_1-\varphi_2)\sin(\varphi_2-\varphi_3)\sin(\varphi_3-\varphi_1)}{4^3 K_1 K_2 K_3}}.
$$
}
\end{lemma}

We also have the inverse of Lemma \ref{4lm:dudphi} at a regular locus:

\begin{lemma} \label{4lm:dudphi3}
For $\varphi_j \in [\varphi_\fb^-, \varphi_\fb^+]$, $(j=1,2,3)$ such that $\varphi_\neq \varphi_j$ $(i\neq j)$, we have
\begin{equation}
\displaystyle{
\begin{pmatrix} d \varphi_1 \\ d \varphi_2 \\ d \varphi_3\end{pmatrix}
=\cK \cM
\begin{pmatrix} d u_1 \\ d u_2\\ du_3\end{pmatrix}
},
\qquad
\cL^{-1}=\cK \cM,
\label{4eq:Elas3.4}
\end{equation}
where
$
\displaystyle{
\cK
:=-
\begin{pmatrix}
\frac{K_1}{\sin(\varphi_2-\varphi_1)\sin(\varphi_3-\varphi_1)}&0& 0 \\
0& \frac{K_2}{\sin(\varphi_3-\varphi_2)\sin(\varphi_1-\varphi_2)}&0 \\
0&0&\frac{K_3}{\sin(\varphi_1-\varphi_3)\sin(\varphi_2-\varphi_3)} 
\end{pmatrix}
}
$ and,
{\small{
$$
\cM
:=
\begin{pmatrix}
8  \sin\varphi_2\sin\varphi_3&
-4\ii(2\ii \sin\varphi_2\sin\varphi_3 - \sin(\varphi_2+\varphi_3) )&
-2 \ee^{-\ii (\varphi_2+\varphi_3)} \\
8  \sin\varphi_1\sin\varphi_3&
-4\ii(2\ii \sin\varphi_1\sin\varphi_3 - \sin(\varphi_3+\varphi_1) )&
-2 \ee^{-\ii (\varphi_1+\varphi_3)} \\
8  \sin\varphi_1\sin\varphi_2&
-4\ii(2\ii \sin\varphi_1\sin\varphi_2 - \sin(\varphi_1+\varphi_2) )&
-2\ee^{-\ii (\varphi_1+\varphi_2)} \\
\end{pmatrix}.
$$
}}
\end{lemma}

\begin{proof}
The straightforward computations show it. \qed
\end{proof}

We remark that (\ref{4eq:Elas3.4}) in Lemma \ref{4lm:dudphi3} means that even if $\varphi_i$ $(i=1,2,3)$ is real, $d\varphi_j$ is complex valued one-form.
We let it decomposed to $d\varphi_j = d\varphi_{j,\rr}+ \ii d \varphi_{j, \ri}$.
Further we introduce $\varphi := \varphi_1 +\varphi_2 +\varphi_3 \in \RR$ and $d\varphi = d\varphi_{\rr}+ \ii d \varphi_{\ri}$; 
$\psi_\rr = 2 \varphi$, $d\psi_\rr = 2 d\varphi_\rr$ and $d\psi_\ri =2 d\varphi_\ri$ for $\psi$ in (\ref{4eq:loopMKdV2}) and (\ref{4eq:gaugedMKdV2}).
We sometimes write $\varpi_{a,\rr}:=\varphi_a$.

\color{black}
\begin{remark}
{\rm{
For a point $\gamma' \in \tX$, the holomorphic one form $\nu(\gamma')$ is regarded as $\nu(\gamma')=\nu(\hkappa_X \gamma')$.
Lemma \ref{4lm:dudphi} means that for $u\in \tv(\gamma=(\gamma_1,\gamma_2, \gamma_3)) = \displaystyle{\sum \int_{\gamma_i} \nu}$, $du = d\tv(\gamma)$ is equal to $\sum \nu(\gamma_i)$.
We regard this as a linear transfomation of $d \gamma_i$, i.e., $du = \cL(d \gamma)$, where $d \gamma= \sum d \gamma_i$ and $d \gamma_i$ is an element of $\hvarpi_x \hkappa_X \tX =\PP^1$ because $d \gamma_i$ is locally defined.
The matrix $\cL$ is the matrix representation of the transfomarion $\cL$.
Since $\cL$ is invertible due to the Abel-Jacobi theorem for the regular locus, $d \gamma = \cL^{-1}(du)= \cK \cM du$ there.
Then the transfomation $\cL^{-1}$ (or the matrix $\cK \cM$) can be also interpreted as the pullback  $\tv^* : T^*_{ u} \CC^3 \to T^*_{\gamma} S^3 \tX$ for a point $\tv(\gamma)=u \in \CC^3$ for the Abelian integral $\tv: S^3 \tX \to \CC^3$, and for the Abel-Jacobi map $\hv: S^3 \hX \to \hJ_X$,
Then we regard $\cL$ as $\tv^{-1*}: T^*_{\gamma} S^3 \tX \to T^*_u\CC^3$.

We note that since $d \gamma_i$ is defined on $\PP^1$ but $\nu(\gamma_i)$ is defined on $\hX$, $\tgamma_i$ is defined for the lift from $\PP^1$ to $\hX$.
Of course, the lift is not a map but we consider $\gamma_i \in \tX$ and thus the sign $\tgamma_i$ is uniquely determined for a given $\gamma_i\in \tX$.

Further, we define these correspondences only for the regular locus, but we show that some of the view can be extended to the global one even for singular locus as we will mention in Section 5.

Originally Weierstrass basically considered these maps to find his sigma function Al-function in \cite{Wei54} by implicitly studying the sine-Gordon equation, and Baker also used this matrix to find the so-called KdV hierarchy in \cite{Baker03}.
They expressed $\partial_{u_i}$ in terms of $\partial_{x_i}$ by using such a matrix, and realized the inversion problem of $\tv$ or $v$ as the Jacobi inversion formula.
}}
\end{remark}

\color{black}

\section{Real hyperelliptic solutions of the gauged MKdV equation over $\RR$}

\revs{
For a point $\gamma \in S^3\tX$ and $u=\tv(\gamma) \in \CC^3$, the fibers $T^*_{ \gamma}S^3 \tX$ and $T^*_{u} \CC^3$ are isomorphic to $\CC^3$ as the real or the complex vector space.
When we regard the matrix $\cK \cM$ as a linear map from $\CC^3$ to $\CC^3$ or the pullback $\tv^*:T^*_{u} \CC^3 \to T^*_\gamma S^3 \tX$  for the point $\gamma$, we decompose the image of the map $\tv^*$ and $\tv^{-1*}$ from the viewpoint of the reality condition so that the decomposition induces a map from $\RR^\ell$ to $\RR^\ell$.
In other words, even though the entries in $\cK \cM$ are complex valued morophic function on $S^3 \hX$, we may find real valued bases $\cV ds$ in $T^*_\gamma S^3 \tX$ and $\be$ in $T^*_\be \CC^3$ so that $\tv^*(\be ds) = \cV ds$.
We provide such basis $\{\cV_1, \cV_2, \cV_3\}$, which is linked to the basis 
$\displaystyle{\left\{\be_1:=\begin{pmatrix} 1/2 \\ -1/2 \\ 1\end{pmatrix}, 
\be_2:=\begin{pmatrix} 1 \\ 0 \\ 0\end{pmatrix}, 
\be_3:=\begin{pmatrix} \ii/2 \\ \ii/2 \\ 0\end{pmatrix}\right\}}$ of $\CC^3$, i.e.,
\begin{equation}
\CC^3 = \langle\be_1, \be_2, \be_3\rangle_\CC, \quad
\RR^2 = \langle\be_1, \be_2\rangle_\RR \subset \CC^3
\label{eq:basis_in_C3}
\end{equation}
by the following lemma.
}

Direct computations show the following lemma:
\begin{lemma}\label{lm4.1}
For the real meromorphic valued vectors on $S^3 \hX$,
\begin{equation*}
\cV_1
:=\begin{pmatrix}
\frac{2\tK_1\cos(\varphi_2-\varphi_3)}{\sin(\varphi_2-\varphi_1)\sin(\varphi_3-\varphi_1)}\\
\frac{2\tK_2\cos(\varphi_3-\varphi_1)}{\sin(\varphi_3-\varphi_2)\sin(\varphi_1-\varphi_2)}\\
\frac{2\tK_3\cos(\varphi_1-\varphi_2)}{\sin(\varphi_1-\varphi_3)\sin(\varphi_2-\varphi_3)}
\end{pmatrix},\quad
\cV_2
:=\begin{pmatrix}
\frac{8\tK_1\sin(\varphi_2)\sin(\varphi_3)}{\sin(\varphi_2-\varphi_1)\sin(\varphi_3-\varphi_1)}\\
\frac{8\tK_2\sin(\varphi_3)\sin(\varphi_1)}{\sin(\varphi_3-\varphi_2)\sin(\varphi_1-\varphi_2)}\\
\frac{8\tK_3\sin(\varphi_1)\sin(\varphi_2)}{\sin(\varphi_1-\varphi_3)\sin(\varphi_2-\varphi_3)}
\end{pmatrix},
\end{equation*}
\begin{equation}
\cV_3
:=\begin{pmatrix}
\frac{2\tK_1\sin(\varphi_2+\varphi_3)}{\sin(\varphi_2-\varphi_1)\sin(\varphi_3-\varphi_1)}\\
\frac{2\tK_2\sin(\varphi_3+\varphi_1)}{\sin(\varphi_3-\varphi_2)\sin(\varphi_1-\varphi_2)}\\
\frac{2\tK_3\sin(\varphi_1+\varphi_2)}{\sin(\varphi_1-\varphi_3)\sin(\varphi_2-\varphi_3)}
\end{pmatrix},\quad
\cC=\displaystyle{
\begin{pmatrix} \tgamma_1 & & \\ &\tgamma_2 & \\ & & \tgamma_3\end{pmatrix}},
\label{eq:g3CIII0}
\end{equation}
where $\tK_j:=\tK(\varphi_j)$, 
we have the following relations,
$$
\revs{\be_1}
=\cL \cC \cV_1, \quad
\revs{\be_2}
=\cL \cC \cV_2, \quad
\revs{\be_3}
=\cL \cC \cV_3.
$$
\end{lemma}

\begin{proof}
We note $\cC^{-1}=\cC$. 
The direct computation show them.
\end{proof}

\color{black}
\begin{remark}\label{rmk:dus_dts}
{\rm{
In terms of $\cV$'s, $\cK\cM$ is expressed by
$$
\cK\cM= (\cV_2, \cV_2+2\ii \cV_3, \cV_1+\cV_2/2+\ii\cV_3), 
$$
In other words, we will decompose the image of the Abelian integral or the Abel-Jacobi map from a real analytic viewpoint or consider the linear transformation in $T^*\CC^3$ to $T^*\CC^3$:
\begin{equation}
\begin{pmatrix} du_1 \\ d u_2 \\ d u_3\end{pmatrix}
=\begin{pmatrix} 1 & \ii/2 & 1/2\\ 0 & \ii/2 & -1/2\\ 0 & 0 & 1
\end{pmatrix}
\begin{pmatrix} dt_1 \\ d t_2 \\ d t_3\end{pmatrix}, \quad
\begin{pmatrix} dt_1 \\ d t_2 \\ d t_3\end{pmatrix}
=\begin{pmatrix} 1 & -1 & -1\\ 0& -2\ii & -\ii \\ 0 & 0 & 1
\end{pmatrix}
\begin{pmatrix} du_1 \\ d u_2 \\ d u_3\end{pmatrix},
\label{eq:dus_dts}
\end{equation}
We let the former matrix denoted by $\cD$, and  also define $dt:=dt_1$ and $ds:=dt_3$. 

Instead of $\cV_1$, we can use the basis $\{\tcV_1, \cV_2, \cV_3\}$ where $\tcV_1:=$ $\displaystyle{\begin{pmatrix}
\frac{2\tK_1\cos(\varphi_2+\varphi_3)}{\sin(\varphi_2-\varphi_1)\sin(\varphi_3-\varphi_1)}\\
\frac{2\tK_2\cos(\varphi_3+\varphi_1)}{\sin(\varphi_3-\varphi_2)\sin(\varphi_1-\varphi_2)}\\
\frac{2\tK_3\cos(\varphi_1+\varphi_2)}{\sin(\varphi_1-\varphi_3)\sin(\varphi_2-\varphi_3)}
\end{pmatrix}=\cV_1+\cV_2/2}$, and consider this problem.
Though we will not show explicitly, we confirmed that the results with $\tcV_1$ have no difference in its nature from the following results based on $\{\cV_1, \cV_2, \cV_3\}$.
}}
\end{remark}
\color{black}

We recall that $\tgamma_i$ and thus $\cC$ represent the sheet from $\hX \to \revs{\PP^1}$.

Since we are interested in the real orbit of $\varphi_a$'s, we treat the real directions given only by $\cC\cV_1d s$ and $\cC\cV_2d t$, but omit the $\cC\cV_3$ direction.
Further, since we are concerned with the orbit including the $u_3$ direction, we will focus on the direction, $\cC\cV_1d s$.

Lemma \ref{lm4.1} leads \revs{to} a nice property which is connected with the reality condition CI and CII in Condition \ref{cnd} by the straightforward computation:
\begin{proposition}\label{4pr:reality_g3}
For configurations $(\ee^{\ii\varphi_i}, \ii8 \tgamma_i \tK_i\ee^{3\ii\varphi_i})_{i=1,2,3} \in S^3 \hX$ such that $\varphi_i \in  [\varphi_\fb^-, \varphi_\fb^+]\subset \RR$ and $\varphi_i \neq \varphi_j$ $(i\neq j)$ for $(d s, d t) \in T^* \RR^2$, one-forms,
\begin{equation}
\begin{pmatrix} d \varphi_{1,\rr} \\ d \varphi_{2,\rr} \\
d \varphi_{3,\rr}\end{pmatrix}
= \cC\cV_1 d s=\begin{pmatrix}
\frac{2\tK_1\tgamma_1\cos(\varphi_2-\varphi_3)}{\sin(\varphi_2-\varphi_1)\sin(\varphi_3-\varphi_1)}\\
\frac{2\tK_2\tgamma_2\cos(\varphi_3-\varphi_1)}{\sin(\varphi_3-\varphi_2)\sin(\varphi_1-\varphi_2)}\\
\frac{2\tK_3\tgamma_3\cos(\varphi_1-\varphi_2)}{\sin(\varphi_1-\varphi_3)\sin(\varphi_2-\varphi_3)}
\end{pmatrix}d s ,
\label{eq:g3CIII}
\end{equation}
and $\cC\cV_2 d t$
\color{black}
form the two-dimensional real subspace in $T^* S^3 \hvarpi_x^{-1}[\varphi_\fb^-, \varphi_\fb^+]$. 
Thus the condition CI is satisfied.

Further $\be_1 ds = \cL\cC (\cC \cV_1) ds$ and $\be_2 dt = \cL\cC (\cC \cV_2) dt$, and thus the condition CII is also satisfied.
\color{black}
\end{proposition}

\begin{proof}
Due to Lemma \ref{lm4.1}, 
\revs{
$\be_1 ds = \cL\cC (\cC \cV_1) ds$ and $\be_2 dt = \cL\cC (\cC \cV_2) dt$.
}
Since the one forms in $\cC\cV_1 d s$ and  $\cC\cV_2 d t$ are holomorphic one-from although $\cV_a$ is singular at the branch point, it is defined over $\hX$.
\end{proof}

Using the relation in Proposition \ref{4pr:reality_g3}, we have real two-dimensional subspaces in $S^3 \hX$ and $J_X$ respectively.
\revs{
In other words, we have the real vector subspace $\displaystyle{\LL_{v_0}:=\left\{v:=\begin{pmatrix} t+s/2\\ -s/2\\ s\end{pmatrix}+v_0 := \int^s \be_1ds + \int^t \be_2dt+v_0 \Bigr|\ (t,s)\in \RR^2\right\}}\subset \CC^3$.
$\LL_{v_0}$ satisfies the condition CII,  i.e., $d u_{2, \ri}=d u_{3,\ri}=0$.
}

\begin{proposition}\label{4pr:reality_g3a}
Assume Assumption \ref{Asmp}.
Let $(P_{a,0}=( \ee^{\ii\varphi_{a,0}},$ $ 8\ii \tgamma_a \tK(\varphi_{a,0})\ee^{3\ii \varphi_{a,0}})))_{a=1,2,3}$ be a point in $S^3 \hX$ where $\varphi_{a,0}\in [\varphi_\fb^-, \varphi_\fb^+]$ such that we avoid $\varphi_{1,0}=\varphi_{2,0}=\varphi_{3,0}$, and $v_0 :=\revs{v(P_{1,0}, P_{2,0}, P_{3,0})}$, i.e., $v_0 = \displaystyle{\sum_{i=a}^3 \int_\infty^{P_{a,0}} \nuI{}}$.
\revs{For $(t,s) \in \RR^2$,}
\begin{equation}
   \revs{\gamma(t,s):=}  \left(\int^s_0 \cC\cV_1 d s + \int^t_0 \cC\cV_2 d t\right) 
\in \revs{S^3\tX},
\label{eq:Xipnt}
\end{equation}
\revs{forms
\begin{equation}
\tSS_{(\varphi_{a,0})}:=\{ \gamma(t,s) \ | \ (t,s) \in \RR^2\},
\label{eq:tS_varphi}
\end{equation}
if exists.
$\hkappa_X(tSS_{(\varphi_{a,0})})$ corresponds a subspace in $S^3 \revs{\Omega^1} \subset S^3 \hX$ such that $\varpi_x \hkappa_X(\gamma(t,s))$ belongs to $[\varphi_\fb^-, \varphi_\fb^+]$.}
Here \revs{$\Omega^1$} is a loop in $\hX$ whose image of $\hvarpi_x:(\ee^{\ii \varphi}, 8\ii K\ee^{3\ii \varphi})\mapsto \ee^{\ii \varphi})$ is the connected arc of the circle displayed in Figure \ref{fg:Fig01} of $[\varphi_\fb^-, \varphi_\fb^+]$. 
\end{proposition}

\begin{proof}
Except the singular locus $\varphi_i = \varphi_j$ $(i\neq j)$, the directions $\cV_1 d s$ and $\cV_2 d t$ are linearly independent and real valued.
Thus for a given $(\varphi_{a,0}) \in S^3 \RR$, $\hvarpi_x\hkappa_X\gamma(t,s)$ belongs to the unit circle or $\varpi_a(t,s)$ ($a=1,2,3$) are real.
Even for the singular locus $\varphi_i = \varphi_j$ $(i\neq j)$, we can define the $\tSS_{(\varphi_{a,0})}$ by Lemma \ref{4lm:intersection}.
\end{proof}

\color{black}
Assume the existence of the subspace in Proposition \ref{4pr:reality_g3a}.
We note $\kappa_\hX(\tSS_{(\varphi_{a,0})})\subset S^3 \Omega^1$..
From the construction, we have $\tv(\gamma(t,s))=\be_1 s +\be_2 t +v_0)$.
Then we consider the graph of $\tv^{-1}$ in $\CC^3\times S^3 \tX$,
\begin{equation}
\cG(\tv^{-1}|\LL_{v_0})=\left\{ \left(\be_1 s +\be_2 t +v_0,\gamma(t,s) \right)\ | \ (t,s) \in \RR^2\right\}\subset \LL_{v_0} \times \tSS_{(\varphi_{a,0})} 
\label{eq:LL_tS}
\end{equation}
by regarding $\LL_{v_0}$ as $\RR^2=\{(t,s)\}$.

\color{black}

\color{black}
In this paper, we consider the following subspaces of $\tSS_{(\varphi_a,0)}$, $\LL_{v_0}$, and $\cG(\tv^{-1}|\LL_{v_0})$ at $t=0$ more precisely.
\begin{equation}
\begin{split}
&\tS_{(\varphi_{a,0}),t}:=\{ \gamma(t,s) \ | \ s \in \RR\},\quad
L_{v_0,t}:=\left\{v:=\begin{pmatrix}s/2\\ -s/2\\ s\end{pmatrix}+v_0 \Bigr|\ s\in \RR\right\}, \\
&\cG(\tv^{-1}|L_{v_0}):=\left\{ \left(\be_1 s +v_0,\gamma(t,s) \right)\ | \ s \in \RR\right\}\subset \tS_{(\varphi_{a,0}),t}\times \LL_{v_0,t}.
\label{eq:tS_varphi,t}
\end{split}
\end{equation}
For a meromorphic function $\psi$ on $S^3 X$, we naturally have its pull-back $\kappa_X$.
Since at every point in $S^3\tX$, the meromorphic function $\psi$ satisfies the MKdV equation over $\CC$ as a differential identity, and $\kappa_\hX \gamma$ for $\gamma \in \tS_{(\varphi_{a,0}),t}$ is real from the definition, we have a real valued one $\psi_{\rr}=2(\varphi_{1}+ \varphi_{2}+\varphi_{3})$.
Then using the graph space $(\gamma, \psi(\kappa_X(\gamma))) \in S^3 X \times \PP$, we have a graph structure from (\ref{eq:Xipnt}),
\begin{equation}
\Psi_t:=\left\{ \left(\be_1 s+v_0, \psi_\rr(\kappa_X(\gamma(t,s))\right)\ | \ s \in \RR\right\}.
\label{eq:Psit}
\end{equation}

\color{black}

More precisely speaking, we have the first theorem of this paper.

\begin{theorem}\label{pr:solgMKdV}
For $t=0$ and  \revs{$\trp(\varphi_{1}(s), \varphi_{2}(s), \varphi_{3}(s)) \in [\varphi_\fb^-, \varphi_\fb^+]^3$ for a solution $\gamma(0,s)$ of} the differential equation (\ref{eq:g3CIII}) for $S^3 \hX$, if exists, we let $\psi_{\rr}=2(\varphi_{1}(s)+ \varphi_{2}(s)+\varphi_{3}(s))$.
Then $\psi_\rr(s)$ is the cross-section
 $\psi_\rr(s)=\psi_\rr(u_2, u_3)|_{u_3=s, u_2=-s/2}$ of $\psi_\rr(u_2, u_3)$ which satisfies the gauged MKdV equation,
\begin{equation}
\revs{(\partial_{u_2}}+\tA(u_2,u_3)\partial_{u_2})\psi_\rr(u_2, u_3)
           +\frac{1}{8}
\left(\partial_{u_3} \psi_\rr(u_2, u_3)\right)^3
+\frac{1}{4}\partial_{u_3}^3 \psi_\rr(u_2, u_3)=0,
\label{4eq:gaugedMKdV2a}
\end{equation}
where the gauge field is $\tA(u_1, u_2)=(\lambda_{6}-3-\frac{3}{4}(\partial_{u_3}\psi_\ri(u_2, u_3))^2)/2$ given by $\partial_{u_3}\psi_\ri(u_2, u_3)=\partial_{u_3} \varphi_{1,\ri}+\partial_{u_3} \varphi_{2,\ri}+\partial_{u_3} \varphi_{3,\ri}$, and
\begin{equation}
\begin{pmatrix} 
\partial_{u_3} \varphi_{1,\ri} \\ 
\partial_{u_3} \varphi_{2,\ri} \\
\partial_{u_3} \varphi_{3,\ri}
\end{pmatrix}
=\begin{pmatrix}
\frac{2K_1\tgamma_1\sin(\varphi_2+\varphi_3)}{\sin(\varphi_2-\varphi_1)\sin(\varphi_3-\varphi_1)}\\
\frac{2K_2\tgamma_2\sin(\varphi_3+\varphi_1)}{\sin(\varphi_3-\varphi_2)\sin(\varphi_1-\varphi_2)}\\
\frac{2K_3\tgamma_3\sin(\varphi_1+\varphi_2)}{\sin(\varphi_1-\varphi_3)\sin(\varphi_2-\varphi_3)}
\end{pmatrix}.
\label{eq:g3CIIIi}
\end{equation}
\end{theorem}

\begin{proof}
\revs{From Proposition \ref{4pr:reality_g3a},} we have the results. Since $du_2 = -d s/2$ and $du_3=d s$, $\psi_\rr(s)=\psi_\rr(u_2, u_3)|_{u_3=s, u_2=-s/2}$.
\revs{
Further, we note $\cK\cM$ is the Jacobi matrix $(\partial \varphi_i/\partial u_j)$, and thus we have $\partial_{u_3} \varphi_{j, \ii}$. 
}
\end{proof}

\section{Global behavior of the hyperelliptic solutions of the gauged MKdV equation over $\RR$}

Based on the above results, we consider the global behavior of the solutions 
$\gamma(0,s)=\displaystyle{\int^s_0 \cV_1 d s}\in S^3 \tX$.
We employ Assumption \ref{Asmp} from here, and for simplicity, we will write it as $\varphi_a(s)=\displaystyle{ \int^s_0 \cV_1 ds}$.

Since we are interested in the case where CIII in Condition \ref{cnd} is satisfied, we have a special look at the initial condition $\partial_{u_3} \psi_{\ri}(s)|_{s=0}=0$.
We note that the \revs{entries} of the matrix $\cC \cK \cM$ consist of $\varphi_1$, $\varphi_2$, and  $\varphi_3$. 
For a certain initial condition $(d\varphi_{1,\ri}, d\varphi_{2,\ri},$ $ d\varphi_{3,\ri})|_{s=0}$ so that $\partial_{u_3} \psi_{\ri}(s)|_{s=0}=0$, we integrate  $\cV_1 d s$..

\subsection{An orbit for the hyperelliptic solution along $s$}
\label{sec:Algorithm}

We let $t=0$.
We will consider the orbit obeying (\ref{eq:g3CIII}) as the differential equation on $S^3 \hX$, $\varphi_a(s)=\displaystyle{\int^s d \varphi_{a,\rr}(s)}$ as $s$ varies from $s=0$ illustrated in Figure \ref{fg:Fig01} (a) and (b).
\begin{theorem}\label{th:solgMKdV_R}
There exists a solution of (\ref{eq:g3CIII}), 
$\varphi_a(s)=\varphi_a(0) +\displaystyle{\int_0^s d \varphi_{a,\rr}(s)}$
for $s\in \RR$, $\varphi_a\in[\varphi_\fb^-,\varphi_\fb^+]$, and 
$((\ee^{\ii\varphi_a},8\ii\tgamma_a \tK(\varphi_a)\ee^{3\ii \varphi_a}))_{a=1,2,3} \in S^3 \hX$.

When we avoid the point $\varphi_1 = \varphi_2= \varphi_3$,
 $\psi_{\rr}(s)=2(\varphi_{1}(s)+ \varphi_{2}(s)+\varphi_{3}(s))$ for $s\in \RR$ is a cross section of the solution of the gauged modified KdV equation (\ref{4eq:gaugedMKdV2a}), i.e., $\psi_\rr(s)=\psi_\rr(u_2, u_3)|_{u_3=s, u_2=-s/2}$ as in Theorem \ref{pr:solgMKdV}.
\end{theorem}

\begin{proof}
We show that a global solution of (\ref{eq:g3CIII}) exists as an orbit in $S^3 \Omega^1$. 
Let us prove this statement for the case of Figure \ref{fg:Fig01} (a) as follows.

\begin{enumerate}

\item Since we are interested in the constant $\partial_{u_3} \psi_{\ri}$, we may set the initial condition $(\varphi_1, \varphi_2, \varphi_3)|_{s=0}$ that $\partial_{u_3} \psi_{\ri}:=2\revs{\partial_{u_3}(\varphi_{1\ri}+ \varphi_{2\ri}+ \varphi_{3\ri})}$ \revs{in (\ref{eq:g3CIIIi})} vanishes.
We must check whether such a state is well-defined.
As in lemma \ref{lmma:s5.5}, it can be assured.

\item We consider a primitive quadrature (\ref{eq:g3CIII}) for a real infinitesimal value $\delta s$.
We find $\delta \varphi_{a,\rr} =\cV_{1,a} \delta s$ for the real part of the entries $\cK \cM$.
Table \ref{1tb:Table1} and the sign of $\cos(\varphi_a + \varphi_b)$ in the numerator in  (\ref{eq:g3CIII}) so that the orbit of $\varphi_a$ moves back and forth between the branch points $[\varphi_\fb^-, \varphi_\fb^+]=\hvarpi_x \Omega^1$ in (\ref{4eq:HEcurve_phi}) as in Figure \ref{fg:Fig01} (a).
Thus $\ee^{\ii \varphi_i}$ exists on the arc of the unit circle as in Figure \ref{fg:Fig01} (a) so that $(\ee^{\ii \varphi_a}, K(\varphi_a))$ draws a loop $\Omega^1$ in $\hX$.

We note that at the branch point, the orbit of $\varphi_a$ turns the direction by changing the sign of $\tgamma_a$ so that $\tgamma_a \sin^2(\varphi_a)\delta \varphi_a/2K_a$ is positive as in Figure \ref{fg:Fig01};
it moves the different leaf of the Riemann surface with respect to the projection $\hvarpi_x:\hX \to \revs{\PP^1}$ after passing the branch points.
Lemmas \ref{4pr:reality_g3b} and \ref{4lm:branch} show that the behavior of (\ref{eq:g3CIII}) and (\ref{eq:g3CIIIi}) are stable at the branch points.

\item Depending on the configurations of $\varphi_1, \varphi_2, \varphi_3$, the sign of $\delta \varphi_{a,\rr}$ and $\partial_{u_g} \varphi_{a,\ri}$ are determined due to (\ref{eq:g3CIII}) and (\ref{eq:g3CIIIi}) respectively.

\item  Lemmas \ref{4lm:intersection} and \ref{4lm:Im_intersection} show that we find an orbit in $S^3 \hX$ as a solution of (\ref{eq:g3CIII}) even at $\varphi_2=\varphi_3$ if $\varphi_1\neq\varphi_3$. Thus (\ref{eq:g3CIIIi}) is also defined there.

\end{enumerate}
It is obvious to show the case of Figure \ref{fg:Fig01} (b) similarly.
\end{proof}

{\small{
\begin{table}[htb]
\caption{The sign of the factor}\label{1tb:Table1}
  \begin{tabular}{|r|c|c|c|}
\hline
 & $\varphi_1>\varphi_2>\varphi_3$& $\varphi_1>\varphi_3>\varphi_2$ 
&$\varphi_2>\varphi_1>\varphi_3$ \\
\hline
$\sin(\varphi_2-\varphi_1)$& $-$ & $-$ & $+$\\
$\sin(\varphi_3-\varphi_1)$& $-$ & $-$ & $-$\\
$\sin(\varphi_3-\varphi_2)$& $-$ & $+$ & $-$\\
$\sin(\varphi_2-\varphi_1)\sin(\varphi_3-\varphi_1)$& $+$ & $+$ & $-$\\
$\sin(\varphi_3-\varphi_2)\sin(\varphi_1-\varphi_2)$& $-$ & $+$ & $+$\\
$\sin(\varphi_1-\varphi_3)\sin(\varphi_2-\varphi_3)$& $+$ & $-$ & $+$\\
\hline
\hline
 & $\varphi_2>\varphi_3>\varphi_1$& $\varphi_3>\varphi_1>\varphi_2$ 
&$\varphi_3>\varphi_2>\varphi_1$ \\
\hline
$\sin(\varphi_2-\varphi_1)$& $+$ & $-$ & $+$\\
$\sin(\varphi_3-\varphi_1)$& $+$ & $+$ & $+$\\
$\sin(\varphi_3-\varphi_2)$& $-$ & $+$ & $+$\\
$\sin(\varphi_2-\varphi_1)\sin(\varphi_3-\varphi_1)$& $+$ & $-$ & $+$\\
$\sin(\varphi_3-\varphi_2)\sin(\varphi_1-\varphi_2)$& $+$ & $+$ & $-$\\
$\sin(\varphi_1-\varphi_3)\sin(\varphi_2-\varphi_3)$& $-$ & $+$ & $+$\\
\hline
  \end{tabular}
\end{table}
}}

Accordingly\revs{,} we check the behavior of the branch point and intersection $\varphi_a=\varphi_b$, and the initial condition as follows.

\subsection{At the branch points}\label{ssec:5.2}

Since the one-forms in $\cC \cV_1 d s$ are holomorphic one-forms, it is trivial that the orbit exists at the branch points.
However\revs{,} it is important to show how the orbit behaves there, and thus we consider it.
We assume $k_1 > k_2 >k_3$ and let $\varphi_\fb:=\varphi_\fb^+$ simply. 
We consider the behavior of the one-forms at the branch point, $\varphi = \pm\varphi_\fb$ here:
\begin{lemma}\label{lm:ss5.2}
Let \revs{$\varphi = \pm(\varphi_\fb-\ft^2)$ and $\tvarphi :=\ft^2$.}
At $\varphi = \pm\varphi_\fb$,
$$
\nuI{1}=\frac{2 \ee^{-2\ii\varphi_\fb}d\ft}{8 K_\fb}, \quad
\nuI{2}=\frac{2\ii\ee^{-\ii\varphi_\fb}\sin(\varphi_\fb)d\ft}{4 K_\fb},
 \quad
\nuI{3}=\frac{-2 \sin^2(-\varphi_\fb)d\ft}{K_\fb}, 
$$
where $K_\fb:=\tgamma\tK_\fb(\varphi)$, 
$$\displaystyle{
\tK_\fb(\varphi):=\frac{\sqrt{
\xi_\fb(t)(1\pm k_1 \sin \varphi)(1-k_2^2 \sin^2 \varphi)
(1-k_3^2 \sin^2 \varphi)}}
{k_1k_2k_3}},
$$
$\revs{\xi_\fb}(t) := \displaystyle{(k_1 \cos\varphi_\fb)+\frac{1}{2!}t^2
-\frac{k_1}{3!}\cos\varphi_c t^4+\cO(t^5)}$ and $\tgamma = \pm 1$.
\end{lemma}

\begin{proof}
$(1\mp k_1\sin\varphi)=(1\mp k_1(\sin\varphi_\fb -\tvarphi\cos\varphi_\fb +\cO(\tvarphi^2))$. $d \tvarphi = -2 \ft d\ft$.
Further, noting $k_1 = \sin\varphi_\fb$, we have
\begin{equation}
\begin{split}
1-k_1 \sin \varphi 
&= 1 - k_1 (\sin \varphi_\fb \cos \ft^2 - \cos \varphi_\fb \sin \ft^2)\\
&= \ft^2 \left(k_\fb\cos \varphi_\fb (1-\frac{1}{3!}\ft^4)
+\frac{1}{2!} \ft^2+\cO(\ft^6)\right).
\end{split}
\end{equation}
Thus we have $\xi_\fb(\ft)$. \qed
\end{proof}
Then the one forms in $\hJ_X$ at the point $\varphi_1 =\pm\varphi_\fb$ are given as follows.

\begin{lemma} \label{4lm:dudphi3b}
For $\varphi_j \in (\varphi_\fb^-,\varphi_\fb^+)$, $(j=2,3)$,  and \revs{$\ft=\sqrt{\varphi_\fb\mp \varphi_1}$}, let $K_j:=K(\varphi_j)$, $j= 2, 3$ at $\varphi_1 =\pm \varphi_\fb$. 
The following holds:

\noindent
$$\displaystyle{
\begin{pmatrix} d u_1 \\ d u_2\\ du_3\end{pmatrix}
=
\begin{pmatrix}
\frac{2 \ee^{-2\ii\varphi_1}}{8 K_\fb}&
\frac{ \ee^{-2\ii\varphi_2}}{8 K_2}&
\frac{ \ee^{-2\ii\varphi_3}}{8 K_3}\\
\frac{2\ii \ee^{-\ii\varphi_1}\sin(\varphi_1)}{4 K_\fb}&
\frac{\ii \ee^{-\ii\varphi_2}\sin(\varphi_2)}{4 K_2}&
\frac{\ii \ee^{-\ii\varphi_3}\sin(\varphi_3)}{4 K_3}\\
\frac{-2 \sin^2(\varphi_1)}{2 K_\fb}&
\frac{- \sin^2(\varphi_2)}{2 K_2}&
\frac{- \sin^2(\varphi_3)}{2 K_3}\\
\end{pmatrix}
\begin{pmatrix} d \ft \\ d \varphi_2 \\d \varphi_3\end{pmatrix}
}.
$$
Let the matrix be denoted by $\cL_\fb$.
\revs{
Then the determinant of $\cL_\fb$ is given by
$$
\det(\cL)=\displaystyle{
\frac{2\sin(\varphi_1-\varphi_2)\sin(\varphi_2-\varphi_3)\sin(\varphi_3-\varphi_1)}{4^3 K_\fb K_2 K_3}}.
$$
}
\end{lemma}

\begin{lemma} \label{4lm:dudphi5}
For $\varphi_j \in (\varphi_\fb^-,\varphi_\fb^+)$, $(j=2,3)$, and $\ft=\sqrt{\mp \varphi_1 -\varphi_\fb}$ at $\varphi_1 =\pm \varphi_\fb$, we have
$$
\displaystyle{
\begin{pmatrix} d \ft_1 \\ d \varphi_2 \\d \varphi_3\end{pmatrix}
=\cK_\fb \cM_\fb
\begin{pmatrix} d u_1 \\ d u_2\\ du_3\end{pmatrix}
},
$$
where
$
\displaystyle{
\cK_\fb
:=
\begin{pmatrix}
\frac{ K_\fb/2}{\sin(\varphi_2-\varphi_1)\sin(\varphi_3-\varphi_1)}&0& 0 \\
0& \frac{K_2}{\sin(\varphi_3-\varphi_2)\sin(\varphi_1-\varphi_2)}&0 \\
0&0&\frac{K_3}{\sin(\varphi_1-\varphi_3)\sin(\varphi_2-\varphi_3)} 
\end{pmatrix}
}
$ and,
{\small{
$$
\displaystyle{
\cM_\fb
:=
\begin{pmatrix}
8  \sin\varphi_2\sin\varphi_3&
-4\ii(2\ii \sin\varphi_2\sin\varphi_3 - \sin(\varphi_2+\varphi_3) )&
-2\ee^{-\ii (\varphi_2+\varphi_3)} \\
8  \sin\varphi_1\sin\varphi_3&
-4\ii(2\ii \sin\varphi_1\sin\varphi_3 - \sin(\varphi_3+\varphi_1) )&
-2 \ee^{-\ii (\varphi_1+\varphi_3)} \\
8 \sin\varphi_1\sin\varphi_2&
-4\ii(2\ii \sin\varphi_1\sin\varphi_2 - \sin(\varphi_1+\varphi_2) )&
-2 \ee^{-\ii (\varphi_1+\varphi_2)} \\
\end{pmatrix}
}.
$$
}}
\end{lemma}

\revs{
\begin{proof}
$\cK_\fb \cM_\fb$ is the inverse matrix of $\cL_\fb$.
\end{proof}
}

We consider the neighborhood at $\varphi_1 =\varphi_\fb$.
For $\varphi_1 < \varphi_\fb$ and $d\varphi_{1,\rr}>0$, $\varphi_1$ moves to $\varphi_\fb$.
Since $\ft^2= \varphi_\fb-\varphi_1$ and we are interested in the real orbit $\ft$, we have $2 \ft d\ft_\rr = d \varphi_{1, \rr}$.
Thus, $\ft$ monotonically increases from $\ft<0$ for $s$. For $\ft=0$, $\varphi_1=\varphi_\fb$ and $\ft$ increases. 
For $\ft>0$, $\varphi_1$ decreases.
In other words, $d \varphi_{1, \rr}$ reverses the direction due to the sign of $K_\fb$, since the sign of $\ft$ corresponds to the sign of $K_\fb=\tgamma_1\tK_\fb$.
We note that the quadrature is done on the hyperelliptic curve $\hX$, and the leaf of $\hX$ as a double cover of $\revs{\PP^1}$ causes this situation.

It is easy to find the following lemma.

\begin{lemma}\label{4pr:reality_g3b}
Corresponding to (\ref{eq:g3CIII}), 
let
\begin{equation}
\begin{pmatrix} d \ft_{1,\rr} \\ d \varphi_{2,\rr} \\
d \varphi_{3,\rr}\end{pmatrix}
=\begin{pmatrix}
\frac{K_\fb\tgamma_1\cos(\varphi_2-\varphi_3)}{\sin(\varphi_2-\varphi_1)\sin(\varphi_3-\varphi_1)}\\
\frac{2K_2\tgamma_2\cos(\varphi_3-\varphi_1)}{\sin(\varphi_3-\varphi_2)\sin(\varphi_1-\varphi_2)}\\
\frac{2K_3\tgamma_3\cos(\varphi_1-\varphi_2)}{\sin(\varphi_1-\varphi_3)\sin(\varphi_2-\varphi_3)}
\end{pmatrix}d s
\label{eq:g3CIIIb}
\end{equation}
we have 
$$
\begin{pmatrix} ds/2 \\ -d s/2\\ d s\end{pmatrix}
=\cL_\fb\, \cC
\begin{pmatrix} d \ft_\rr \\ d \varphi_{2,\rr} \\
d \varphi_{3,\rr}\end{pmatrix}.
$$
\end{lemma}

\begin{lemma}\label{4lm:branch}
The quadrature of (\ref{eq:g3CIII}) is well defined at the branch point $\varphi_\fb^+$ and $\varphi_\fb^-$.
\end{lemma}

\subsection{Intersection: real part}\label{ssec:5.3}

We consider the behavior at the intersection of two orbits $\varphi_a(s)$ and $\varphi_b(s)$, $(a\neq b)$ but three points $\varphi_a$ $(a=1,2,3)$ do not coincide simultaneously.
We note that the divisor of the meromorphic function $\tx-x$ as a function with respect to $\tx$ is $(x,y)+(x,-y)$.
The Abel theorem shows that the Abelian integral of $(x,y)+(x,-y)$ must lie in periodic lattice points \cite{FarkasKra}.  
The intersection means that we consider the integral $\displaystyle{\int_{(x,y)}^{(x,-y)} \nuI{}}$, which must be the value associated with the period of the lattice in the Jacobi variety due to the Abel theorem.
It may be considered as {\lq}an algebraic collapse\rq.
Hence the intersection in $S^3 X$ is crucial.
Even though it generates the collapse as $\partial_{s} \psi_\rr$, the behavior of $\psi_\rr$ is not so bad, and the solution of (\ref{eq:g3CIII}) is defined as an orbit in $[\varphi_\fb^-, \varphi_\fb^+]^2$ as in the following results.

\begin{lemma}\label{4lm:intersection}
The quadrature of (\ref{eq:g3CIII}) as the differential equation for $\hX$ is also defined at the crossing point $\varphi_a = \varphi_b$ for $a\neq b$, i.e.,
we define the orbit in $s$-$(\varphi_1, \varphi_2, \varphi_3)$-space as $\RR \times S^3 [\varphi_\fb^-, \varphi_\fb^+]$ unless $\varphi_1=\varphi_2=\varphi_3$.
\end{lemma}

\begin{proof}
Let us assume that two orbits $(\varphi_1, K_1>0)$ with the positive direction $d\varphi_{1,\rr}$ and $(\varphi_2, K_2<0)$ with the negative direction $d\varphi_{2,\rr}$ intersect at $\varphi_0$ and $s=s_0$ as $\hvarpi_x : \hX \to \revs{\PP^1}$, i.e.,
$d \varphi_{1,\rr}>0$, $d \varphi_{2,\rr} < 0$, and $\varphi_2> \varphi_1$ for $s< s_0$.
We consider the intersection point $\varphi_0$ in $(s,\varphi)$-plane.
Then we introduce the real local parameters $\eta_1$ and $\eta_2$ such that $\varphi_1 = \varphi_0 + \eta_1$, $\eta_1 \in (-\varepsilon, \varepsilon)$ and $\varphi_2 = \varphi_0 - \eta_2$, $\eta_2 \in (-\varepsilon, \varepsilon)$ for $1 \gg \varepsilon > 0$, i.e., $d\eta_1>0$ and $d \eta_2>0$.
Let $-\tgamma_1=\tgamma_2=1$, and thus we have
$K_1 = K_0 + \partial_{\varphi}K_0 \eta_1 + o(\eta_1)$ and
$K_2 = -K_0 + \partial_{\varphi}K_0 \eta_2 + o(\eta_2)$, where
$\displaystyle{K':=\frac{\partial K(\varphi)}
{\partial \varphi}}$ is equal to
$$
-\frac{\sin(2\varphi)(3(k_1k_2k_3)^2\sin(\varphi)^4 - 2(k_1^2k_2^2+k_1^2k_3^2+k_2^2k_3^2)\sin(\varphi)^2+ (k_1^2 + k_2^2 + k_3^2))}{2K(\varphi)}.
$$
Hence $\displaystyle{\begin{pmatrix}
d \varphi_{1,\rr}\\
d \varphi_{2, \rr}
\end{pmatrix}}$ is given by 
\begin{gather*}
\begin{split}
\begin{pmatrix}
d \eta_1\\
-d \eta_2
\end{pmatrix}
&=\begin{pmatrix}
\frac{-K_1\cos(\varphi_2-\varphi_3)d s}{\sin(\varphi_2-\varphi_1)\sin(\varphi_3-\varphi_1)}\\
\frac{-K_2\cos(\varphi_1-\varphi_3)d s}{\sin(\varphi_1-\varphi_2)\sin(\varphi_3-\varphi_2)}\\
\end{pmatrix}\\
&=\left(\begin{matrix}
\frac{-(K_0+K'_0\eta_1)\cos(\varphi_0-\eta_2-\varphi_3)d s}{-\sin(\eta_2+\eta_1)\sin(\varphi_3-\varphi_0-\eta_1)} +d_{>0}(\eta_1, \eta_2)\\
\frac{(K_0-K_0'\eta_2)\cos(\varphi_0+\eta_1-\varphi_3)d s}{\sin(\eta_1+\eta_2)\sin(\varphi_3-\varphi_0+\eta_2)}+d_{>0}(\eta_1, \eta_2)\\
\end{matrix}\right)\\
&=\left(\begin{matrix}
\frac{(K_0\cos(\varphi_0-\varphi_3)+K'_0\cos(\varphi_0-\varphi_3)\eta_1+K_0\sin(\varphi_0-\varphi_3)\eta_2) d s}
{(\eta_2+\eta_1)\sin(\varphi_3-\varphi_0)}(1+\frac{\cos(\varphi_3-\varphi_0)\eta_1}{\sin(\varphi_3-\varphi_0)})\\
-\frac{(K_0\cos(\varphi_0-\varphi_3)-K_0'\cos(\varphi_0-\varphi_3)\eta_2-K_0\sin(\varphi_0-\varphi_3)\eta_1)d s}{(\eta_1+\eta_2)\sin(\varphi_3-\varphi_0)}
(1-\frac{\cos(\varphi_3-\varphi_0)\eta_2}{\sin(\varphi_3-\varphi_0)})\\
\end{matrix}\right. \\
&\quad\qquad \left.\begin{matrix}
 +d_{>0}(\eta_1, \eta_2)\\
+d_{>0}(\eta_1, \eta_2)\\
\end{matrix}\right)\\
\end{split}
\end{gather*}
\begin{gather*}
\begin{split}
&=\left(\begin{matrix}
\frac{K_0\cos(\varphi_0-\varphi_3)d s}{(\eta_2+\eta_1)\sin(\varphi_3-\varphi_0)}
\left(1+
\left(\frac{K'_0\cos(\varphi_0-\varphi_3)}{K_0\cos(\varphi_0-\varphi_3)}
+\frac{\cos(\varphi_3-\varphi_0)}{\sin(\varphi_3-\varphi_0)}\right)\eta_1
+\frac{\sin(\varphi_0-\varphi_3)}{\cos(\varphi_0-\varphi_3)}\eta_2\right) \\
-\frac{K_0\cos(\varphi_0-\varphi_3)d s}{(\eta_1+\eta_2)\sin(\varphi_3-\varphi_0)}
\left(1-
\left(\frac{K_0'\cos(\varphi_0-\varphi_3)}{K_0\cos(\varphi_0-\varphi_3)}
+\frac{\cos(\varphi_3-\varphi_0)}{\sin(\varphi_3-\varphi_0)}\right)\eta_2
-\frac{\sin(\varphi_0-\varphi_3)}{\cos(\varphi_0-\varphi_3)}
\eta_1
\right)\\
\end{matrix}\right. \\
&\quad\qquad \left.\begin{matrix}
 +d_{>0}(\eta_1, \eta_2)\\
+d_{>0}(\eta_1, \eta_2)\\
\end{matrix}\right).\\
\end{split}
\end{gather*}
Here $d_{>\ell}(t_1, t_2)$ denotes an element in the formal power series $\CC[[t_1, t_2]]$ whose smallest degree is $\ell+1$, and similarly $d_{>\ell}(t)$ for $\CC[[t]]$.
They can be expressed as
\begin{equation}
\begin{split}
&(\eta_1+\eta_2)(1-\fb_1 \eta_1-\fb_2 \eta_2 +d_{>0}(\eta_1, \eta_2))
d\eta_1= -\fa d s,\\
&(\eta_1+\eta_2)(1+\fb_1 \eta_2+\fb_2 \eta_1 +d_{>0}(\eta_1, \eta_2))
d\eta_2= -\fa d s,
\end{split}
\label{4eq:deta}
\end{equation}
where
$$
\fa = \frac{K_0 \cos(\varphi_0-\varphi_3)}{\sin(\varphi_0-\varphi_3)}(>0),
\quad
\fb_1=\frac{K'_0}{K_0}+\frac{\cos(\varphi_3-\varphi_0)}{\sin(\varphi_3-\varphi_0)},
\quad
\fb_2=\tan(\varphi_0-\varphi_3)
$$
$$
\frac{d \eta_2}{d \eta_1}= (1-2 \fb_1 \eta_1 - 2 \fb_2\eta_2) 
+d_{>0}(\eta_1, \eta_2)).
$$
We note $\eta_2(0)=0$.
We substitute the expansion $\eta_2= \eta_1+\fc_1 \eta_1^2 +d_{>2}(\eta_1)$ into the equation,
$$
1 + 2\fc_1 \eta_1 = 1 -2 \fb_1 \eta_1 -2 \fb_2 \eta_1 +d_{>1}(\eta_1).
$$
We have
$\eta_2 = \eta_1 -(\fb_1 + \fb_2) \eta_1^2+ d_{>2}(\eta_1)$.
Thus the behavior of the $s$ at the crossing point $s_0$,
\begin{gather}
s-s_0= -\frac{1}{\fa_1}\eta_1^2+ \frac{(\fb_1+\fb_2)}{\fa_1} \eta_1^3 + d_{>3}(\eta_1).
\label{5eq:17}
\end{gather}

We note that $\eta_2$ is defined as the function of $\eta_1$ for $(-\varepsilon, \varepsilon)$.
Further since we consider $\varphi_1,\varphi_2$ as a point in the symmetric product $S^2 [\varphi_\fb^-, \varphi_\fb^+] \subset S^3 [\varphi_\fb^-, \varphi_\fb^+]$ rather than $ [\varphi_\fb^-, \varphi_\fb^+]^2$, we can exchange $\eta_1$ and $\eta_2$ freely.
Thus, since for $s < s_0$, we consider $\varphi_1<\varphi_2$, for $s>s_0$, we consider $\varphi_1 > \varphi_2$ by swapping these roles or swapping each direction.
Although $\partial_{u_3} \varphi_{i,\rr}$ seems to diverge at the point $s=s_0$ due to $\displaystyle{\frac{d \eta_i}{d s}=\frac{\fa_1}{2\eta_1}(1+d_{>0}(\eta_1))}$, we can find the orbit in $s$-$(\varphi_1, \varphi_2)$ as $\RR \times S^2 [\varphi_\fb^-, \varphi_\fb^+]$.
Even for the case, we have
$$
d\varphi_{1,\rr} +d\varphi_{2,\rr} = d (\eta_1-\eta_2) = [2(\fb_1 + \fb_2) \eta_1+d_{>1}(\eta_1, \eta_2)] d \eta_1.
$$
$d (\varphi_{1,\rr} +\varphi_{2,\rr}) /d s $ diverges at the point, we can define the orbit as a curve in $s-(\varphi_{1,\rr} +\varphi_{2,\rr})$ plane.
\end{proof}

\subsection{Intersection: imaginary part}\label{ssec:5.4}

Similarly, we consider the case of the imaginary part $\partial_{u_3} \varphi_{a,\ri}$ using the same setting as in the previous subsection.

\begin{lemma}\label{4lm:Im_intersection}
For the quadrature of (\ref{eq:g3CIII}), $\psi_{\ri}$ is also defined at the crossing point $\varphi_a = \varphi_b$ for $a\neq b$.
\end{lemma}

\begin{proof}
\begin{gather*}
\begin{split}
\begin{pmatrix}
\partial_{u_3} \varphi_{1,\ri}\\
\partial_{u_3} \varphi_{2,\ri}
\end{pmatrix}
&=\begin{pmatrix}
\frac{-K_1\sin(\varphi_2+\varphi_3)}{\sin(\varphi_2-\varphi_1)\sin(\varphi_3-\varphi_1)}\\
\frac{K_2\sin(\varphi_1+\varphi_3)}{\sin(\varphi_1-\varphi_2)\sin(\varphi_3-\varphi_2)}\\
\end{pmatrix}\\
&=\begin{pmatrix}
\frac{-(K_0+K'_0\eta_1)\sin(\varphi_0-\eta_2+\varphi_3)}{-\sin(\eta_2+\eta_1)\sin(\varphi_3-\varphi_0-\eta_1)} +d_{>0}(\eta_1, \eta_2)\\
\frac{(-K_0+K_0'\eta_2)\sin(\varphi_0+\eta_1+\varphi_3)}{\sin(\eta_1+\eta_2)\sin(\varphi_3-\varphi_0+\eta_2)}+d_{>0}(\eta_1, \eta_2)\\
\end{pmatrix}.
\end{split}
\end{gather*}

By substituting $\eta_2(\eta_1)$ and $s(\eta_1)$ in the proof of Lemma \ref{4lm:intersection}into the relation, we have
\begin{equation}
\partial_{u_3}\varphi_{1,\ri} +\partial_{u_3}\varphi_{2,\ri} = 0+\eta_1+d_{>1}(\eta_1).
\label{eq:dpsi_i}
\end{equation}
It vanishes at the cross point $\eta_1=0$.
\end{proof}

\subsection{Initial condition}\label{ssec:5.5}

We obtain the configuration of $(\varphi_1, \varphi_2, \varphi_3)$ such that $\partial_{u_3} \varphi_\ri(\varphi_1, \varphi_2, \varphi_3)$ vanishes as an initial condition as follows.

\begin{lemma}\label{lmma:s5.5}
Assume that $\varphi_1=\varphi_{\fb}$ such that $K_1=K(\varphi_{\fb})=0$,
$\partial \psi_{\ri}/\partial s$ is equal to
$$
-\frac{K_2\sin(\varphi_3+\varphi_\fb)}{\sin(\varphi_3-\varphi_2)\sin(\varphi_\fb-\varphi_2)}
+\frac{K_3\sin(\varphi_\fb+\varphi_2)}{\sin(\varphi_\fb-\varphi_3)\sin(\varphi_2-\varphi_3)}=0,
$$
whose solution is $(\phi_2, K_2) =(\phi_3,K_3)$ due to (\ref{eq:dpsi_i}).
Then $\partial_{u_3} \psi_{\ri}$ vanishes at $s=0$.
\end{lemma}

\section{Discussion}

We obtain real cross sections of solutions of the gauged MKdV equation (\ref{4eq:gaugedMKdV2a}).
\revs{
By noting $\cD^{-1} =(\partial t_i/\partial_j)$ in  (\ref{eq:dus_dts}), we have the correspondence of the differential operators,
\begin{equation}
\begin{pmatrix} \partial_{u_1} \\ \partial_{u_2} \\ \partial_{u_3}\end{pmatrix}
=\begin{pmatrix} 1 & 0 & 0\\ -1 & -2\ii & 0\\ -1 & -\ii & 1
\end{pmatrix}
\begin{pmatrix} \partial_{t} \\ \partial_{t_2} \\ \partial_{s}\end{pmatrix}.
\label{eq:pus_pts}
\end{equation}
Thus the cross sections does not correspond to the solutions of (\ref{4eq:gaugedMKdV2a}) directly but it must pick up the properties of the real solutions.
Hence we go on to discuss the properties of the cross section as follows.
}

If it satisfies the condition CIII in Condition \ref{cnd}, we have a cross sections of the hyperelliptic solution of the real MKdV equation (\ref{4eq:rMKdV2}).

It means the following theorem:
\begin{theorem}\label{4th:reality_g3}
$\psi_{\rr}:=2(\varphi_1+ \varphi_2+\varphi_3)$ of the quadrature $d \varphi_{i, \rr}$ $(i=1,2,3)$ of (\ref{eq:g3CIII}) is a cross section of a local solution of the MKdV equation (\ref{4eq:rMKdV2}) for the region in which $\partial_{u_3} \psi_\ri$ is constant.
\end{theorem}

\bigskip

Although the above theorem can hold if $\partial_{u_3} \psi_\ri=$constant number $c\in \RR$, we focus on the case where $c=0$.

Here we note that (\ref{eq:g3CIII}) holds for every point in $\varphi$'s, but each sign $\tgamma$ in (\ref{eq:3.3_1stdiff}) is determined by the configurations of $(\varphi_1, \varphi_2, \varphi_3)\in S^3 [\varphi_\fb^-, \varphi_\fb^+]$, so the orbit in $S^3 \hX$ proceeds \revs{or we obtain $\gamma \in \S^3 \tX$.}
We weaken the vanishing $\partial_{u_3} \psi_\ri$ condition and replace it with the condition that the maximum of $\partial_{u_3} \psi_\ri$ is much smaller than the maximum of $\partial_s \psi_\rr$, which we check numerically.
\revs{
To estimate the magnitude of $\partial_{u_3} \psi_{\ri}$, we introduce $\displaystyle{
\psi_\ri^\circ:=\int^s \partial_{u_3} \psi_{\ri} ds}$.
}
\begin{figure}
\begin{center}
\includegraphics[width=0.44\hsize, bb= 0 0 448 273]{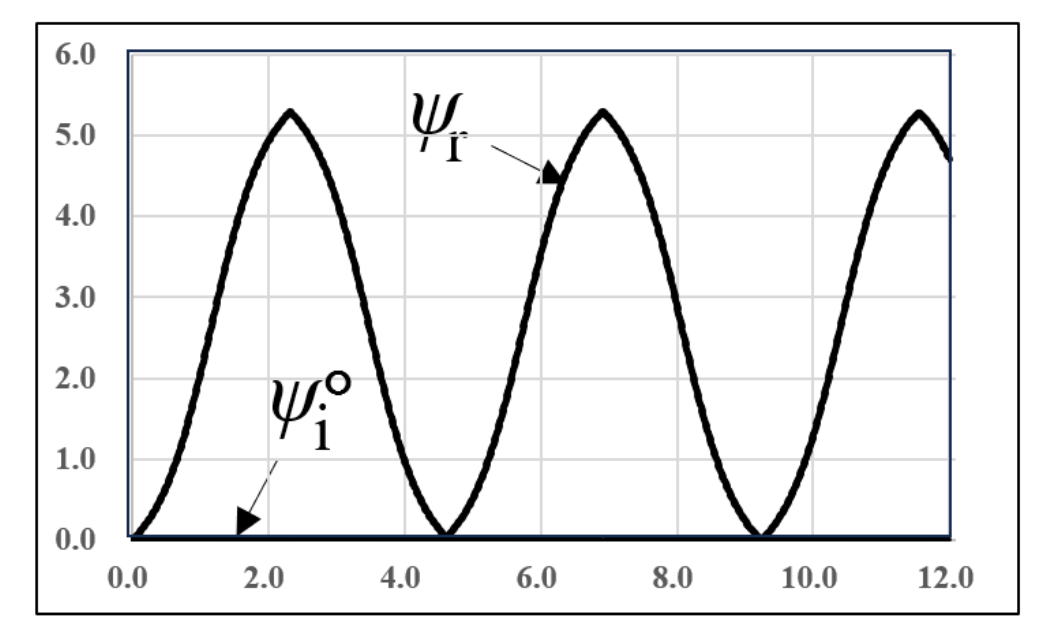}
\hskip 0.1\hsize
\includegraphics[width=0.44\hsize, bb=  0 0 446 273]{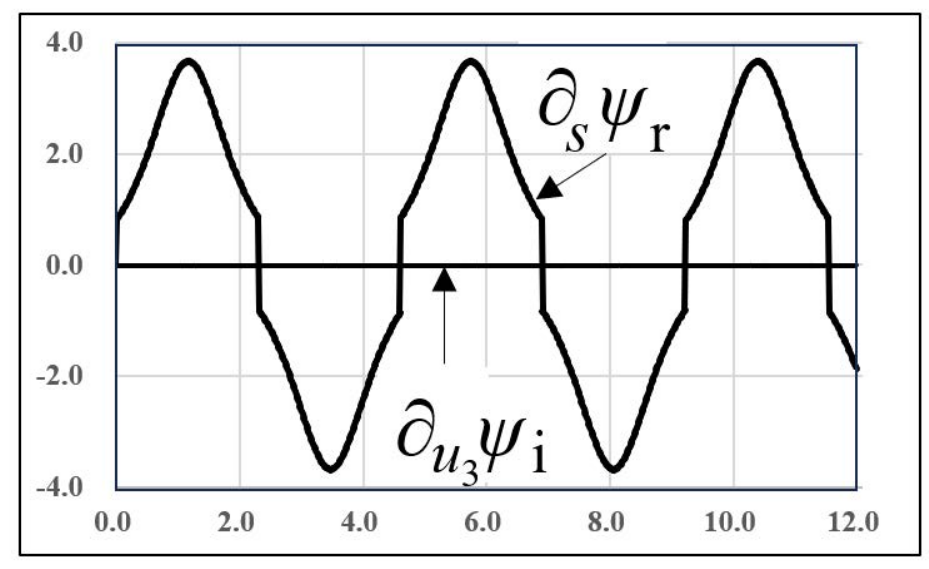}

(a) \hskip 0.35\hsize (b)
\end{center}

\caption{
The solutions of \revs{(\ref{eq:Xipnt}) for $t=0$}:
$(k_1, k_2, k_3) = (1.0400, 1.0392, 1.010)$ and the initial condition is 
$(\varphi_1, \varphi_2, \varphi_3) = (\varphi_\fb, -0.90, -0.90)$
(a): $\psi_\rr/2$ and $\varphi_{a}$, $(a = 1, 2, 3)$.
 and
(b): $\partial_s \psi_\rr$ and $\partial_{u_3} \psi_\ri$ \revs{(\ref{eq:g3CIIIi})}.
}\label{fg:shape01}
\end{figure}

\revs{
We numerically integrate the differential equation (\ref{eq:g3CIII}) following the algorithm introduced in Subsection \ref{sec:Algorithm}.
In other words, we have used the Euler's numerical quadrature method \cite{RLeV} for $\{(s, \hkappa_X\gamma(0,s))
\ |\ s \ge 0\}\subset \RR \times S^3 \tX$ and 
$\{(s, \psi_\rr((\gamma(0,s)))
\ |\ s \ge 0\}\subset \RR \times \RR$,
In other words, we will draw some graphs $(s, \psi_\rr(s))$, $(s, \psi^\circ_\ri(s))$ and others based on the graph structure $\Psi_t$ in (\ref{eq:Psit}).
}

The first result is displayed in Figure \ref{fg:shape01}.
For the hyperelliptic curve given by $(k_1, k_2, k_3) = (1.0400, 1.0392, 1.010)$, we put the initial condition $(\varphi_1, \varphi_2, \varphi_3) = (\varphi_\fb, -0.90, -0.90)$.
Figure \ref{fg:shape01} (a) displays these $\psi_\rr$ and $\psi_\ri^\circ$ in Figure \ref{fg:shape01} (b), the maximum of the absolute value $|\partial_{u_3} \psi_\ri|$ is $1.00952\times 10^{-2}$ that is 
much smaller than that of $|\partial_s \psi_\rr|$, $3.66894$.
In other words, the orbit is regarded as a cross section of the MKdV equation (\ref{4eq:rMKdV2}) rather than the gauged MKdV equation (\ref{4eq:gaugedMKdV2}).

The second result is illustrated in Figure \ref{fg:shape02}.
We used the hyperelliptic curve given by $(k_1, k_2, k_3) = (1.0260, 1.0259, 1.0008)$.
The initial condition is set by $(\varphi_1, \varphi_2, \varphi_3) = (\varphi_\fb, -0.90, -0.90)$.
Figure \ref{fg:shape02} (a) shows these $\varphi_{a}$ and $\psi_\rr/2$.
As in Figure \ref{fg:shape02} (b) and (c) show $\psi_\rr$, $\psi^\circ_\ri$. $\partial_s \psi_\rr$ and $\partial_{u_3} \psi_\ri$.
The value of $\partial_{u_3} \psi_\ri$ cannot be negligible.
However, the orbit in a certain interval that the variation of $\partial_{u_3} \psi_\ri$ is relatively small could be regarded as one of the MKdV equations (\ref{4eq:rMKdV2}).

\begin{figure}
\begin{center}

\includegraphics[width=0.50\hsize, bb=0 0 494 368]{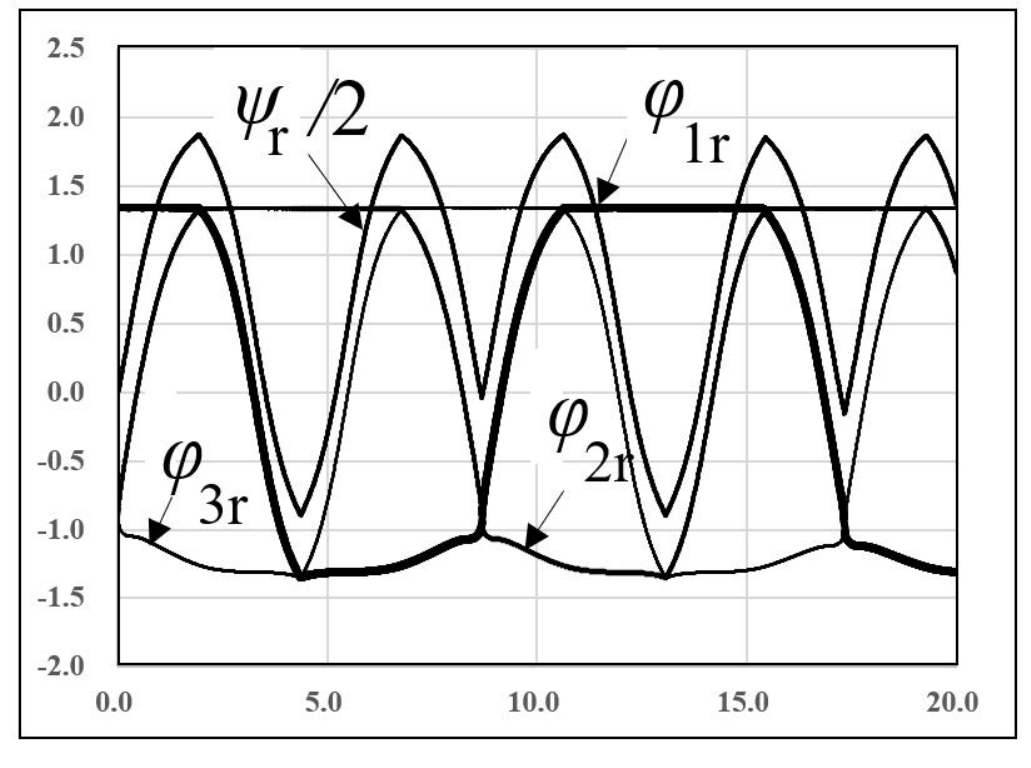}

(a)

\includegraphics[width=0.44\hsize, bb=0 0 465 316]{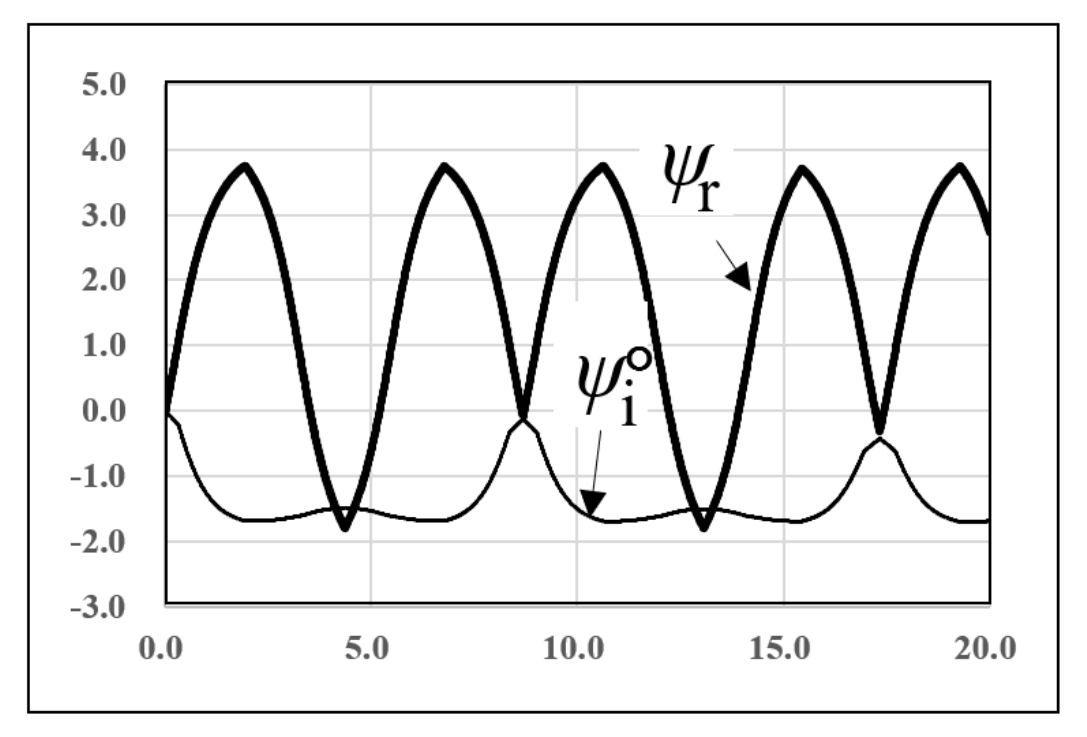}
\hskip 0.1\hsize
\includegraphics[width=0.44\hsize, bb=0 0 506 363]{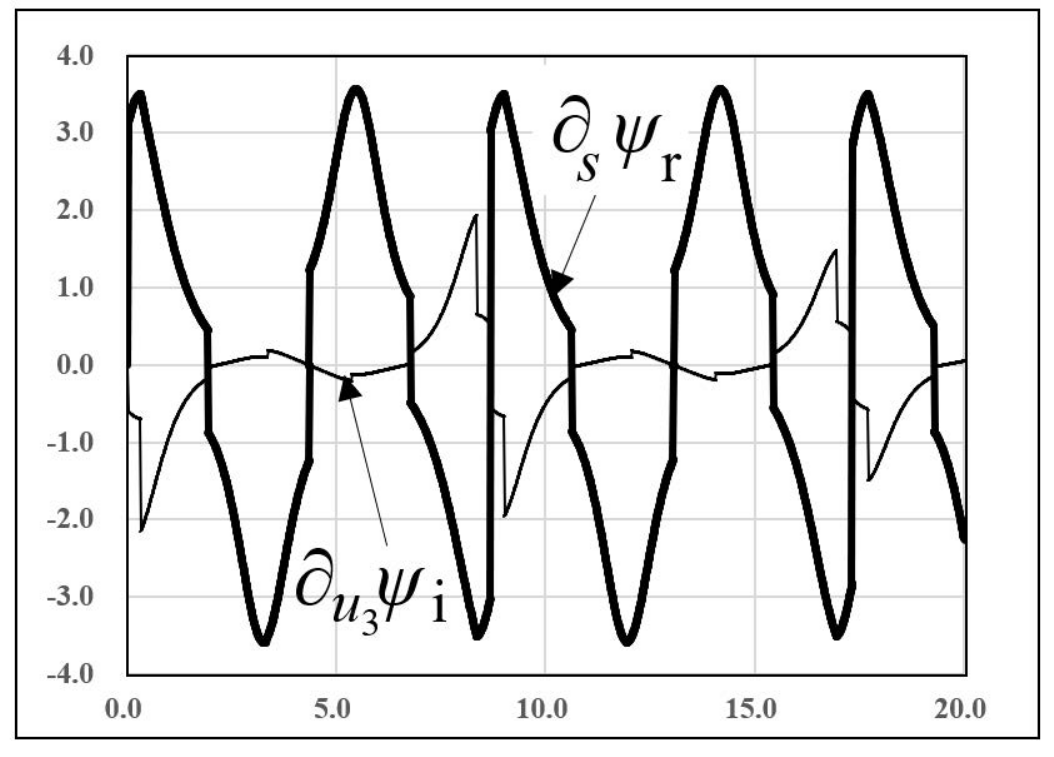}

(b) \hskip 0.35\hsize (c)

\end{center}

\caption{
\revs{The solutions of (\ref{eq:Xipnt}) for $t=0$}::
$(k_1, k_2, k_3) = (1.0260, 1.0259, 1.0008)$ and the initial condition is 
$(\varphi_1, \varphi_2, \varphi_3) = (\varphi_\fb, -0.90, -0.90)$.
(a): $\psi_\rr/2$ and $\varphi_{a}$, $(a = 1, 2, 3)$.
(b): $\partial_s \psi_\rr$ and $\partial_{u_3} \psi_\ri$.
(c) (d): $\partial_s \varphi_{i \rr}$ and $\partial_{u_3} \varphi_{i, \ri}$, $(i = 1, 2, 3)$.
}\label{fg:shape02}
\end{figure}

\section{Conclusion}

In this paper, we showed a novel algebro-geometric method to obtain the cross-section $\psi_\rr(s)$ of the solutions of the focusing gauged MKdV equation (\ref{4eq:gaugedMKdV2}) as in Theorem \ref{pr:solgMKdV} associated with the focusing MKdV equation over $\CC$.
Although on the orbits of $\psi_\rr(s)$, we encounter the branch points of $\varphi_a$ and the intersections between $\varphi_a$ and $\varphi_b$, which might look ill-posed, Theorem \ref{th:solgMKdV_R} guarantees that certain solutions of (\ref{eq:Xipnt}) are globally defined for $s \in \RR$.

Then we have used the data of the hyperelliptic curves $X$ directly  instead of  the Jacobian $J_X$. 
We note that our algebraic study of the algebraic curves on two decades \cite{BEL97b, BuMi2, Mat01a, MP22, M24} allows the such treatment.

Further, Theorem \ref{4th:reality_g3} shows that they are also regarded as \revs{cross section of} the solutions of the MKdV equation (\ref{4eq:rMKdV2}) if the $\partial_{u_3} \psi_{\ri}$ is sufficiently small.
We consider that the question of what the conditions should be for it to be small is the next issue.

\revs{
The ultimate purpose of this study is to find the focusing MKdV equation of higher genus explicitly rather than its cross section as mentioned in Introduction.
It is far from the final result, but this study shows the first step towards the goal.}
In other words, finding the real analytical solution of the focusing MKdV equation is difficult but so interesting that we should continue to study this problem.
As in \cite{Mat97, MP22, M24a}, when we obtain it, it clarifies the shapes of supercoiled DNA.
\revs{
Certainly the recent results based on this study show that the shapes of the generalized elastica reproduce the typical properties of the supercoiled DNAs observed in the laboratory \cite{M24a}.
}

\bigskip
\bigskip

\noindent
{\bf{Acknowledgment}:}
This project was started with Emma Previato 2004 in Montreal and had been collaborated until she passed away June 29, 2022.
Though the author started to step to genus three curves without her, he appreciate her contributions and suggestions which she gave him to this project during her lifetime.
Further, it is acknowledged that John McKay who passed way April 2022 invited the author and her to his private seminar in Montreal 2004 since he considered that this project \cite{Mat97} must have been related to his Monster group problem \cite{McKay, MP16}.
Thus this study is devoted to Emma Previato and John McKay.
The author thanks to Junkichi Satsuma, Takashi Tsuboi, and Tetsuji Tokihiro for inviting him to the Musashino Center of Mathematical Engineering Seminar and for valuable discussions and to Yuta Ogata, Yutaro Kabata and Kaname Matsue for helpful discussions and suggestions.
\revs{
The author is grateful to the anonymous reviewers for their helpful and valuable suggestions; he thanks for the suggestion that (\ref{4eq:rMKdV2_Ishi}) should not called the MKdV equation.
}
He also acknowledges support from the Grant-in-Aid for Scientific Research (C) of Japan Society for the Promotion of Science, Grant No.21K03289.

\end{document}